\documentclass[copyright,creativecommons,noderivs,noncommercial]{eptcs}
\providecommand{\event}{Gabriel Ciobanu (Ed.): Membrane Computing and\\ Biologically Inspired Process Calculi 2009} 
\providecommand{\volume}{15}   
\providecommand{\anno}{2009}   
\providecommand{\firstpage}{1} 
\providecommand{\eid}{1}       

\usepackage{amsmath,amsfonts,amsthm,amssymb,latexsym}
\usepackage{stmaryrd,proof}
\usepackage{hyperref}
\usepackage{enumitem}
\usepackage{graphicx}
\usepackage{pgf,tikz}
\usepackage{biobigraphs}
\usepackage{version}
\usetikzlibrary{decorations}
\usetikzlibrary{decorations.pathreplacing}
\usetikzlibrary{arrows}
\usetikzlibrary{patterns}
\usetikzlibrary{calc}

\newcommand{\defeq}{\triangleq}

\newcommand{\cat}[1]{{\normalfont\textsc{#1}}}

\newcommand{\drb}[1]{\llbracket #1 \rrbracket}
\newcommand{\dsb}[1]{\llparenthesis\, #1 \,\rrparenthesis}
\newcommand{\rlts}[2]{\frac{\displaystyle #1}{\displaystyle #2}}
\newcommand{\xnrightarrow}[1]{\xrightarrow{#1}\kern-3ex \not\kern+3ex}
\newcommand{\xnRightarrow}[1]{\xRightarrow{#1}\kern-3ex \not\kern+3ex}
\newcommand{\set}[1]{\{#1\}}
\newcommand{\grow}{\rhd}

\newcommand{\kcomp}{\mathop{,}}
\newcommand{\mext}{\textrm{\sf m}^\textrm{\sf ext}}
\newcommand{\mcis}{\textrm{\sf m}^\textrm{\sf cys}}
\newcommand{\pcarry}{\mathsf{p}^\mathsf{c}}
\newcommand{\pmemb}{\mathsf{p}^\mathsf{m}}
\newcommand{\pdir}{\mathsf{p}^\mathsf{d}}
\newcommand{\fcell}{\mathsf{f}^\mathsf{c}}
\newcommand{\fmemb}{\mathsf{f}^\mathsf{m}}
\newcommand{\fdir}{\mathsf{f}^\mathsf{d}}
\newcommand{\leftcompadj}{\mathbin{\circ\kern-0.35ex-}}
\newcommand{\rightcompadj}{\mathbin{-\kern-0.35ex\circ}}
\newcommand{\lefttensoradj}{\mathbin{\otimes\kern-0.35ex-}}
\newcommand{\righttensoradj}{\mathbin{-\kern-0.35ex\otimes}}

\newcommand{\sepupto}[1]{\stackrel{#1}{\otimes}}
\newcommand{\biosig}{\mathcal{K}_B}
\newcommand{\orivec}{\vec}
\newcommand{\xrightarrowtriangle}[1]
{\stackrel{#1}{{-\kern-1.05ex\rightarrowtriangle}}}
\newcommand{\Fp}{\cat{F}_{P}}
\newcommand{\Fm}{\cat{F}_{M}}

\newcommand{\rl}[2]{{\displaystyle\frac{#2}{#1}}}
\newcommand{\up}{\shortuparrow\kern-0.5ex}
\newcommand{\down}{\shortdownarrow\kern-0.5ex}

\tikzstyle{membranecolor}=[fill=yellow!50, draw=yellow!60!gray]
\tikzstyle{membranefill}=[fill=yellow!50]
\tikzstyle{membranedraw}=[draw=yellow!60!gray]

\newtheoremstyle{theorem}
{\topsep}
{\topsep}
{}
{}
{\bfseries}
{.}
{0.5em}
{\thmname{#1}\thmnumber{ #2}\thmnote{ (#3)}}%

\newtheoremstyle{definition}
{}
{}
{}
{}
{\bfseries}
{.}
{0.5em}
{\thmname{#1}\thmnumber{ #2}\thmnote{ (#3)}}%

\theoremstyle{theorem} 
\newtheorem{proposition}{Proposition}[section]

\newtheorem{theorem}{Theorem}[section]

\theoremstyle{definition}
\newtheorem{definition}{Definition}[section]

\theoremstyle{remark}

\title{Bigraphical models for protein and membrane interactions}
\author{Giorgio Bacci
\institute{University of Udine}
\email{giorgio.bacci@dimi.uniud.it}
\and
Davide Grohmann
\institute{University of Udine}
\email{grohmann@dimi.uniud.it}
\and
Marino Miculan
\institute{University of Udine}
\email{miculan@dimi.uniud.it}
}

\def\titlerunning{Bigraphical models for protein and membrane interactions}
\def\authorrunning{G.~Bacci \& D.~Grohmann \& M.~Miculan}

\excludeversion{extended}\includeversion{shortversion}

\begin{document}

\maketitle


\begin{abstract}
  We present a \emph{bigraphical framework} suited for modeling
  biological systems both at protein level and at membrane level.  We
  characterize formally bigraphs corresponding to biologically
  meaningful systems, and bigraphic rewriting rules representing
  biologically admissible interactions.  At the protein level, these
  bigraphic reactive systems correspond exactly to systems of
  $\kappa$-calculus.  Membrane-level interactions are represented by
  just two general rules, whose application can be triggered by
  protein-level interactions in a well-defined and precise way.



  This framework can be used to compare and merge models at different
  abstraction levels; in particular, higher-level (e.g. mobility)
  activities can be given a formal biological justification in terms
  of low-level (i.e., protein) interactions.  As examples, we
  formalize in our framework the vesiculation and the phagocytosis
  processes.




\end{abstract}



\section{Introduction}





\looseness=-1
Cardelli in \cite{cardelli05:amsb} has convincingly argued that the
various \emph{biochemical toolkits} identified by biologists can be
described as a hierarchy of \emph{abstract machines},
each of which can be modelled using methods and techniques from
concurrency theory. These machines 
are highly interdependent: ``to understand the functioning of a cell,
one must understand also how the various machines interact''
\cite{cardelli05:amsb}.  Like other complex situations, it seems
unlikely to find a single notation covering all aspects of a whole
organism.  In fact, we are in presence of a \emph{tower of models}
\cite{milner09:tower}, each focusing on specific aspects of the
biological system, at different levels of abstractions. Higher-level
models must be represented, or \emph{realised}, at a lower level, and
where possible this representation must be proved sound; in addition,
we need to combine different models at the same level.
To this end, we need a general \emph{meta}model, that is, a
\emph{framework}, where these models (possibly at different
abstraction levels) can be encoded, and their interactions can be
formally described.


In this paper, we substantiate Milner's idea that \emph{bigraphs} can
be successfully used as a framework for systems biology.  More
precisely, we define a class of \emph{biological bigraphs}, and
\emph{biological bigraphical reactive systems (BioRS)}, for dealing
with both protein-level and membrane-level interactions.

An important design choice is that this framework has to be
\emph{biologically sound}, i.e., it must admit only systems and
reactions which are biologically meaningful, especially at lower level
machines (i.e. protein).  In this way, encoding a given model,
for any abstract machine, as a BioRS provides automatically a
formal, biologically sound justification for the model (or
``implementation'') in terms of protein reactions and explains how its
membrane-level interactions are realised by protein machinery.

In order to formalize this ``biological soundness'', we need a
\emph{formal} protein model to compare to our framework. We choose
Danos and Laneve's \emph{$\kappa$-calculus}, one of the most accepted
formal model of protein systems.  By suitable sorting conditions, we
define a bigraphical framework which allows \emph{all} and \emph{only}
protein configurations and interactions of the $\kappa$-calculus.  It
is important to notice, however, that our methodology is general, and
can be applied to other formal protein models.

On the other hand, membrane nesting reconfiguration can be performed
by just only two general rules, corresponding to the natural phenomena
of ``pinch'' and ``fuse'' \cite{cardelli08:tcs}.  For encoding a given
membrane model one has just to refine this general schema by
specifying \emph{when} these reactions are triggered, that is, when
the right proteins are in place.  Indeed, as observed by biologists
(see \cite{au:phago} and \cite[Ch.~15]{alberts:ecb}), membrane
interactions present always a ``preparation phase'', where membrane
proteins (receptors and ligands) interact, followed by the actual
``membrane reconfiguration'' phase.  The ``pinch'' and ``fuse'' rules
model the latter phase, whilst the preparation phase depends on the
specific proteins involved and hence left to the encoding of the
specific model under examination.

Another important consequence of encoding a model in our framework is
that any ``too abstract'', or non-realistic, aspect of the model is
readily identified, because its formalization turns out to be
problematic, or even impossible. In this case, one has to change the
model, or the encoding, in order to be biologically sound. As a
result, the bigraphical encoding of a given system may exhibit
unexpected features, not present (or unpredicted) in the original
model.  Far from being a problem, this allows to foresee
\emph{emerging properties}, such as behaviours due to interactions of
different abstraction levels, and which cannot be observed within a
single machine model due to its intrinsic abstractions.

A further motivation for our framework comes from the many general
results provided by bigraphs. We mention here only the construction of
compositional bisimilarities \cite{milner:ic06}, allowing to prove
that two systems are \emph{observational equivalent}, that is, they
can be exchanged in any organism without that the overall behaviour
will change. Also for this application it is important to restrict the
bigraphs allowed by the framework to only those biologically
meaningful.


We summarise briefly our approach.  
In Section~\ref{sec:plg} we first define \emph{protein link graphs}
and corresponding reactive systems, which correspond precisely to
protein solutions and protein transition systems of the
$\kappa$-calculus (recalled in Section~\ref{sec:kappa}).
%
%
In Section~\ref{sec:biobig} we extend this approach to deal with
compartments, introducing \emph{biobigraphs}.  Membranes are
represented by two new nodes, which play no role at the protein level,
but which 
can contain other systems.  Appropriate sorting conditions will
enforce biological properties such as bitonality and orientation of
membranes.
Two applications of this frameworks are given in
Section~\ref{sec:examples}: a formal description of vesicle formation
process, and a formalization of the Fc receptor-mediated phagocytosis.
In both cases, the framework obliges to provide a formal justification
of (membrane mobility) reactions in terms of protein-level
interactions.
Conclusions are in Section~\ref{sec:concl}.

\section{Preliminaries}\label{sec:prelim}

\subsection{Bigraphical reactive systems}\label{sec:bigraphs}
In this section we recall the bigraphical framework
\cite{milner:ic06}, extending the variant of \cite{bs:ppdp06} with
typed names.

\smallskip

\noindent\textbf{Bigraphs}
A bigraph represents an open system, so it has an inner and an outer
interface to ``interact'' with subsystems and the surrounding
environment.  The \emph{numbers} in the interfaces describe the
\emph{roots} in the outer interface (that is, the various locations
where the nodes live) and the \emph{sites} in the inner interface
(that is, the grey holes where other bigraphs can be inserted). On the
other hand, the \emph{names} in the interfaces describes the open
links, that is end points where links from the outside world
can be pasted, creating new links among nodes.
An example of a bigraph is shown in Fig.~\ref{fig:bigex}.
\begin{figure}[t]
  \centering
  \includegraphics[scale=0.5]{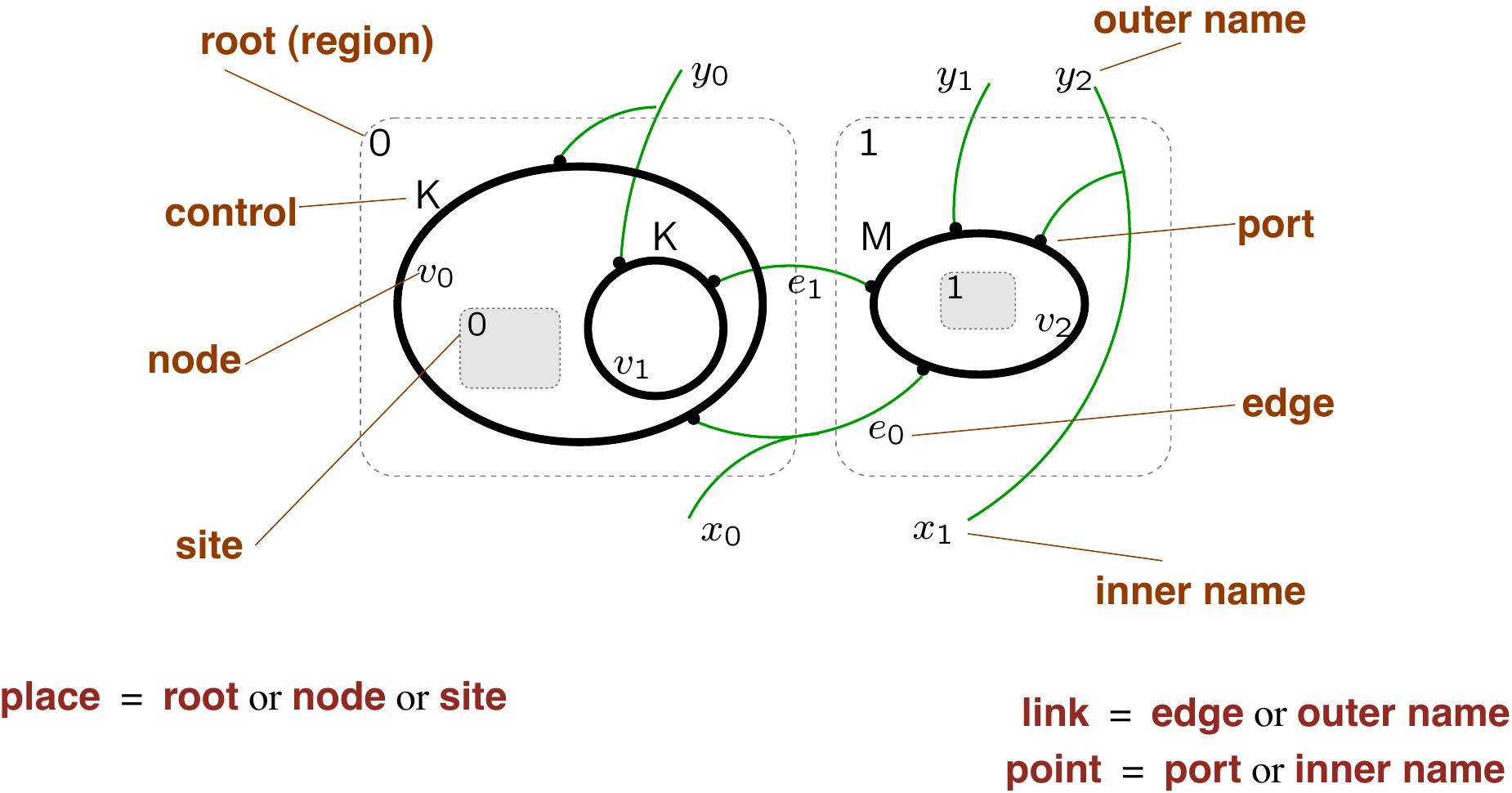}
  \caption{Example of a bigraph and its structure (picture from
    Milner's slide).}
  \label{fig:bigex}
\end{figure}

A signature is a quadruple $\langle \mathcal{K}, ar, \mathcal{T},
\mathcal{E} \rangle$ where $\mathcal{K}$ is a set of \emph{node
  controls} with an arity function \mbox{$ar\colon \mathcal{K} \to
  \mathbb{N}$}, $\mathcal{T}$ is a poset of \emph{types} with ordering
$\sqsubseteq$, and $\mathcal{E} \subseteq \mathcal{T}$ denotes the
\emph{edge types}.

\begin{definition}[Interfaces and Bigraphs]\label{def:bigraphs}
  An \emph{interface} is a pair $\langle n, X \rangle$, where $n$ is a
  finite ordinal (called \emph{width}) and $X$ is a set of typed names
  $x{:\,}\mathsf{t}$, where $\mathsf{t} \in \mathcal{T}$.

  A \emph{bigraph} $G$ is composed by a \emph{place graph} $G^P$,
  describing the nesting of nodes, and a \emph{link graph} $G^L$,
  describing the (hyper-)links among nodes.
  \begin{align}
    G^P & = (V, ctrl, prnt) \colon m \to n
    \tag{Place graph} \\
    G^L & = (V, E, ctrl, edge, link) \colon X \to Y
    \tag{Link graph} \\
    G & = (G^P,G^L) = (V, E, ctrl, edge, prnt, link) \colon
    \langle m,X\rangle \to\langle n,Y\rangle
    \tag{Bigraph}
  \end{align}
  where $V,E$ are the sets of nodes and edges respectively,
  $ctrl\colon V \to \mathcal{K}$ is the \emph{control map}, which
  assigns a control to each node, $edge\colon E \to \mathcal{E}
  (\subseteq \mathcal{T})$ is the \emph{edge typing} which assigns a
  type to each edge, $prnt\colon m\uplus V \to V\uplus n$ is the
  (acyclic) \emph{parent map}, the disjoint sum $Prt = \sum_{v \in V}
  ar(ctrl(v))$ is the set of ports (associated to all nodes), and
  $link\colon X\uplus Prt \to E\uplus Y$ is the (type consistent)
  \emph{link map}, i.e., for all $x{:\,}\mathsf{t} \in X$, if
  $link(x{:\,}\mathsf{t})= l{:}\mathsf{t}'$ then
  $\mathsf{t}\sqsubseteq \mathsf{t}'$.
\end{definition}

In the following, we say \emph{place} to mean a node, a root or a
site, and we call the elements in the domain of the link map
\emph{points}, while the ones in the codomain \emph{links}.  An
\emph{idle} name is not a target of any point.  We denote by $(v,i)$
the $i$th port of the node $v$.  Two nodes are \emph{peer} if two of
their port (one per node) have the same target link.  Often, we say
control to mean a node with associate that control, and we write $x \in X$ 
in place of $x{:\,}\mathsf{t} \in X$ when the type of the name it is clear
from the context.

\begin{definition}[Bigraph category]
  Given a signature $\langle \mathcal{K}, ar, \mathcal{T}, \mathcal{E}
  \rangle$, the category of (link typed) bigraphs
  $\cat{Big}(\mathcal{K},\mathcal{T}, \mathcal{E})$ has interfaces as
  objects and bigraphs as morphisms.
  
  Given two bigraphs $G\colon\langle m,X\rangle \to \langle
  n,Y\rangle$, $H\colon\langle n,Y\rangle \to \langle k,Z\rangle$, the
  composition $H\circ G\colon \langle m,X\rangle \to \langle
  k,Z\rangle$ is defined by composing their place and link graphs:
  \begin{description}[font=\bfseries, labelindent=\parindent,
    leftmargin=\parindent]
  \item[Place:] $H^P \circ G^P = (V, ctrl, (id_{V_G}\cup prnt_H) \circ
    (prnt_G\cup id_{V_H}))\colon m \to k$;
  \item[Link:] $H^L\circ G^L= (V, E, ctrl, edge, (id_{E_G}\cup link_H)
    \circ (link_G\cup id_{P_H}))\colon X \to Z$,
  \end{description}
  where $V = V_{G} \cup V_{H}$, $E = E_{G} \cup E_{H}$ (supposing $V_G
  \cap V_H = E_G \cap E_H = \emptyset$ otherwise a renaming is
  applied), $ctrl = ctrl_{G} \cup ctrl_{H}$ and $edge = edge_{G} \cup
  edge_{H}$.

  $\cat{Plg}$ and $\cat{Lnk}$ will denote the categories of place and
  link graphs, respectively.
\end{definition}
An example of splitting of a bigraph in its components is shown in
Figure~\ref{fig:bigsplit}.
\begin{figure}[t]
  \centering
  \includegraphics[scale=0.4]{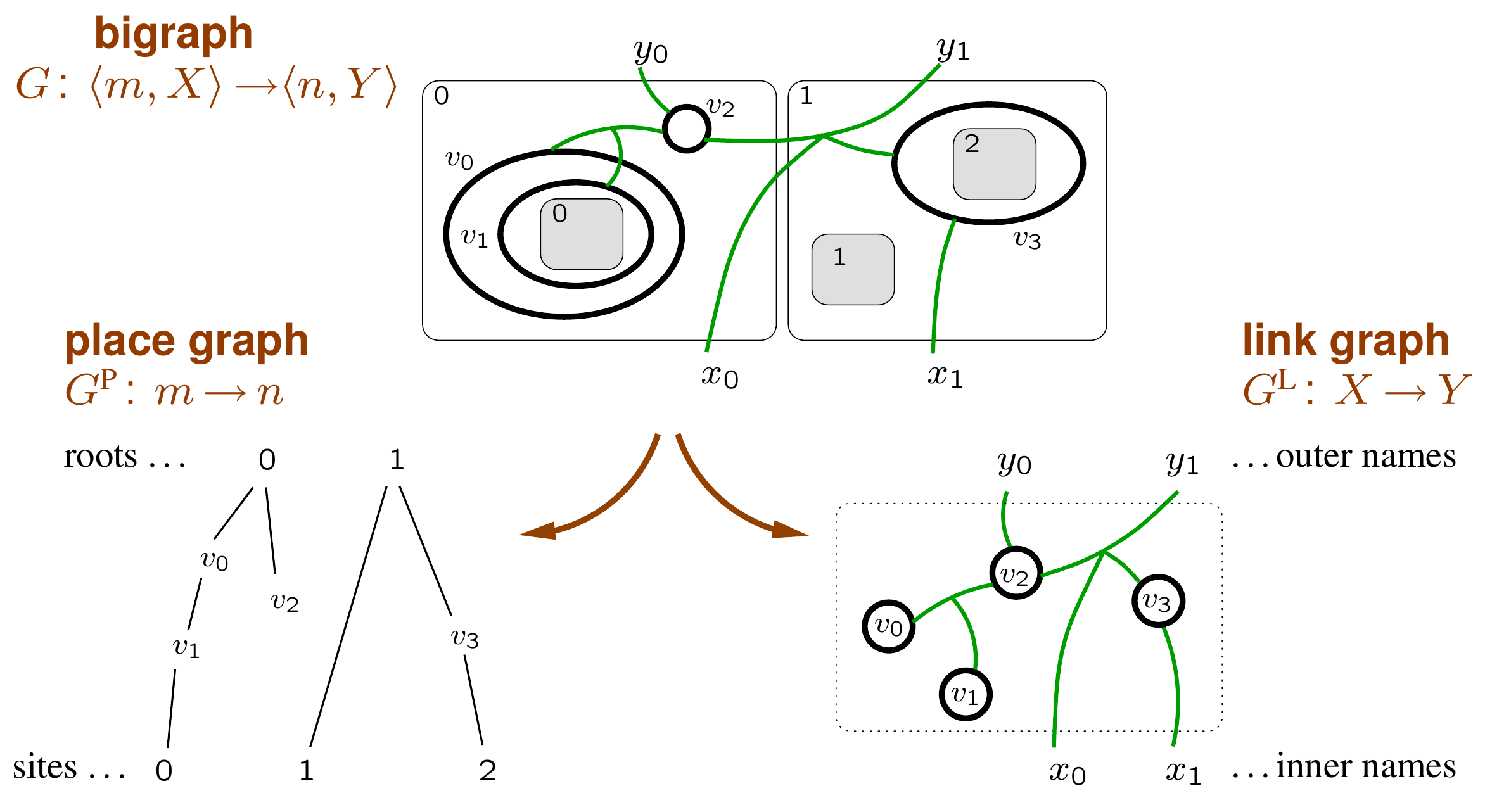}
  \caption{Splitting of a bigraph (picture from \cite{milner:ic06}),
    where $m = 2$, $n = 1$, $X = \set{x_{0}, x_{1}}$, $Y = \set{y_{0},
      y_{1}}$.}
  \label{fig:bigsplit}
\end{figure}

Intuitively, composition is performed in two steps. First, the place
graph are composed by putting the roots of the ``lower'' bigraph
inside the sites (i.e., holes) of the ``upper'' one, respecting the
order given by the ordinal in the interface; then, the links are
connected by sticking together the end parts of connections in the two
link graphs, labelled with the same typed name in the common
interface.

An important operation on bigraphs is the \emph{tensor product}.
Intuitively, the tensor product of $G:\langle m,X\rangle \to \langle
n,Y\rangle$ and $H:\langle m',X'\rangle \to \langle n',Y'\rangle$, 
defined when $X\cap X' = Y\cap Y' =\emptyset$, is
the bigraph $G\otimes H : \langle m+m',X\uplus X'\rangle \to \langle
n+n',Y\uplus Y'\rangle$ and is obtained by putting ``side by side'' $G$ and $H$.
Two useful variant of tensor product can be defined using tensor and
composition: the \emph{parallel product} $\parallel$, which merges
shared names between two bigraphs, and the \emph{prime product}
$\mid$, that moreover merges all roots in a single one.  We refer to
\cite{milner:ic06} for formal definitions.



\smallskip

\noindent\textbf{Bigraphical reactive systems} (BRSs) are reactive
systems in the sense of \cite{lm:concur00}, built over the category
$\cat{Big}(\mathcal{K},\mathcal{T}, \mathcal{E})$
\cite{milner:ic06,bs:ppdp06}. A BRS consists of a set of bigraphical
rewriting rules, together with a definition of the evaluation contexts
(i.e., bigraphs with holes) where rules redexes can be found in order
to be rewritten.  Such contexts are defined as a sub-category
$\mathcal{A}$ of \emph{active context} inside
$\cat{Big}(\mathcal{K},\mathcal{T}, \mathcal{E})$.

\begin{definition}[Activity]
  Each control of a signature is either \emph{active} or \emph{passive} 
  or \emph{atomic}. An atomic control cannot contain any node (hence
  it is a leaf on the place graph).
  A location (i.e., an element in the domain of $prnt$) is
  \emph{active} if all its ancestors have active controls (roots
  are always active); otherwise it is \emph{passive}.
  A context is \emph{active} if all its holes are active locations.
\end{definition}
It is easy to check that active contexts form a compositional
reflective sub-category of bigraphs, which we denote by $\mathcal{A}$,
the category of active contexts.

\begin{definition}[BRS]
  A \emph{bigraphical reactive system}
  $\mathcal{D}(\mathcal{K},\mathcal{T},\mathcal{E},\mathcal{R})$ is
  formed by $\cat{Big}(\mathcal{K},\mathcal{T},\mathcal{E})$ equipped
  with the subcategory $\cat{Act}$ of active contexts, and a set
  $\mathcal{R}$ of (parametric) \emph{reaction rules}, that is pairs
  $L,R: \langle m,X\rangle \to \langle n,Y\rangle$ (usually written as
  $L\rightarrowtriangle R$).
  The \emph{reaction relation} $\rightarrowtriangle$ defined by a BRS
  is the relation between ground bigraphs given by the following
  rule\footnote{Milner's definition of BRS allows for non-linear
    instantiations, but for the scope of this paper linear rules are
    sufficient.}:
  \begin{equation*}
    \frac{C \in \cat{Act} \qquad L \rightarrowtriangle R \in \mathcal{R}
      \qquad d \text{ ground} \qquad
      G = C \circ L\circ d \qquad H = C \circ R\circ d }
    {G \rightarrowtriangle H}
  \end{equation*}
\end{definition}

Clearly, these definitions can be applied to place and link graphs
too.  Active contexts for place graphs are as for bigraphs and the
rules can have only holes; instead for link graphs, all contexts are
active.

\smallskip

\noindent\textbf{Sortings}
Usually, bigraphs defined over a given signature are ``too many''.  A
systematic way for ruling out unwanted bigraphs is by means of a
\emph{sorting} \cite{bdh:concur08}, that is, a functor $F\colon
\cat{S} \to \cat{Big}$, \emph{faithful} and \emph{surjective on
  objects}, defined on a category $\cat{S}$ where unwanted morphisms
from $\cat{Big}$ are deleted.  A sorting can be conveniently defined
by means of a suitable logical condition specifying the class of
morphisms we want to restrict to \cite{bdh:concur08}.  To ease
readability, in the paper we will present these conditions in a
semi-formal logical language; a formal description of all
these conditions can be given as BiLog formulas \cite{cms:icalp05}%
\begin{extended} 
(see Appendix~\ref{sec:formulae})
\end{extended}.
Notice that also the name typing introduced in
Definition~\ref{def:bigraphs} can be expressed by means of a sorting.



\subsection{The $\kappa$-calculus}\label{sec:kappa}
Here we recall Danos and Laneve's $\kappa$-calculus~\cite{dl:kappa04}, 
a calculus for protein interactions, which will be our reference protein-level 
model in Section~\ref{sec:plg}.

Let $\mathcal{P}$ be a finite set of proteins and $\mathcal{N}$ an
enumerable set of edge names. Let $s\colon \mathcal{P} \to
\mathbb{N}$ be a signature, that assigns to every protein the number
of its \emph{domain sites}.
A $\kappa$-interface is a map from $\mathbb{N}$ to $\mathcal{N} \uplus
\set{h,v}$ (ranged over $\rho$, $\sigma$, $\ldots$). Given an
interface $\rho$ and a protein name $A \in \mathcal{P}$, 
a site $(A,i)$ is \emph{visible} if $\rho(i) = v$,
\emph{hidden} if $\rho(i) = h$, and \emph{tied} if $\rho(i)\in
\mathcal{N}$. A site is \emph{free} if it is visible or hidden.  In
the following, we will write $\rho = 1 + \bar{2} + 3^x$ to mean
$\rho(1) = v$, $\rho(2) = h$, and $\rho(3) = x$.
The syntax of $\kappa$-solutions is the following:
\begin{equation*}
  S,T ::= \mathbf{0} \mid A(\rho) \mid S \kcomp T \mid (x)(S)
\end{equation*}
up-to $\alpha$-equivalence and the least structural equivalence
($\equiv$), satisfying the abelian monoid on composition and the scope
extension and extrusion laws. The sets of free edge names $fn(\rho)$ and
$fn(S)$, on interfaces and solutions, are defined as usual.

\begin{definition}
  The set of \emph{connected $\kappa$-solutions} is defined
  inductively as:
  \begin{equation*}
    \rlts{}{A(\rho)} \qquad
    \rlts{S}{(x) S} \qquad 
    \rlts{S\quad T\quad fn(S)\cap fn(T)\neq\emptyset}{S \kcomp T} \qquad
    \rlts{S\quad S\equiv T}{T}
  \end{equation*}
\end{definition}

\begin{definition}
  A $\kappa$-solution $S$ is \emph{(strong) graph-like} if free names
  occur at most (exactly) two times in $S$, and binders tie either 0
  or 2 occurrences.

  A \emph{complex} is a closed, graph-like, and connected solution.
\end{definition}

Biochemical reactions are either complexations (i.e.~where low energy
bonds are formed on two complementary sites) or decomplexations,
possibly co-occur\-ring with (de)activation of sites.  \emph{Causality}
does not allow simultaneous complexations \emph{and} decomplexations
on the same site: this biological constraint is assured in the
$\kappa$-calculus by the notion of \emph{(anti-)monotonicity}, which
forces reactions not to decrease (increase) the level of connection of
a solution (see \cite{dl:kappa04} for details).

\begin{definition}[Growing relation]
  Let $\tilde{x}$ be a set of fresh names.

  A growing relation $\grow$ over $\kappa$-interfaces is defined
  inductively as follows:
  \begin{equation*}
    \rlts{}{\tilde{x} \vdash \bar{i} \grow i} \quad \;
    \rlts{}{\tilde{x} \vdash i \grow \bar{i}} \quad \;
    \rlts{x\in \tilde{x}}{\tilde{x} \vdash i \grow i^x} \quad \;
    \rlts{\tilde{x}\cap fn(\rho)=\emptyset}{\tilde{x}\vdash\rho\grow\rho}
    \quad \;
    \rlts{\tilde{x}\vdash\rho\grow\sigma\quad\tilde{x}\vdash\rho'\grow\sigma'
    }{\tilde{x}\vdash \rho+\rho'\grow\sigma+\sigma'}
  \end{equation*}

  A growing relation $\grow$ over $\kappa$-solutions is defined
  inductively as follows:
  \begin{align*}
    \rlts{}{\tilde{x} \vdash \mathbf{0}\grow\mathbf{0}} \qquad
    \rlts{\tilde{x}\vdash S\grow T\quad\tilde{x}\vdash\rho\grow\sigma
    }{\tilde{x}\vdash S\kcomp A(\rho)\grow T\kcomp A(\sigma)} \qquad
    \rlts{\tilde{x}\vdash S\grow T\quad fn(\sigma)\subseteq\tilde{x}
    }{\tilde{x}\vdash S \grow T\kcomp A(\sigma)}
  \end{align*}
\end{definition}

\begin{definition}[Monotone reactions]
  Let $L,R$ be two $\kappa$-solutions. $L \rightarrow (\tilde{x})R$ is
  a \emph{monotone reaction} if $\tilde{x}\vdash L\grow R$, $L$ and
  $(\tilde{x})R$ are graph-like and $R$ is connected.  $L \rightarrow
  (\tilde{x})R$ is an \emph{antimonotone reaction} if $R \rightarrow
  (\tilde{x})L$ is monotone.
\end{definition}

Finally, we are able to characterize a protein transition system:
\begin{definition}[Protein transition systems]
  Given a set $\mathcal{R}_\kappa$ of monotone and antimonotone
  reactions, we define a \emph{protein transition system} (PTS) as a
  pair $\langle \mathcal{S},\leadsto\rangle$, such that $\mathcal{S}$
  is a set of $\kappa$-solutions, and $\leadsto$ is the least binary
  relation over $\mathcal{S}$ such that $\mathcal{R}_\kappa\subseteq
  {\leadsto}$, closed w.r.t.~$\equiv$, composition and name
  restriction.
\end{definition}


\section{Protein link graphs}\label{sec:plg}
In this section we introduce a general graphical formalism for
modelling biological protein interactions. Proteins are represented
as atomic nodes where links stand for chemical bonds. 
In such a model locations are not needed, hence we do not make
use of the full bigraphical framework and link graphs are enough.


A protein changes its behaviour depending on its folding structure
which determines its current interacting interface (the set of its
domain sites).  This structure could change after a complexation
reaction.  To describe the interface of each protein we use four linking
types, $\mathsf{v}$ (visible), $\mathsf{h}$ (hidden), $\mathsf{b}$ (bond)
and $\mathsf{f}$ (free), indicating the status of the corresponding protein 
site which they are connected to.

\begin{definition}[Protein signature]
  Let $\mathcal{P}$ be a set of protein controls, a \emph{protein
    signature} $\langle \mathcal{P},ar \rangle$ is defined as $\langle
  \mathcal{P},ar,\mathcal{T}_{p}, \mathcal{E}_{p}\rangle$, where
  $\mathcal{T}_{p} =
  \set{\mathsf{h},\mathsf{v},\mathsf{b},\mathsf{f}}$ with
  $\sqsubseteq_{p} = id_{\mathcal{T}_{p}} \cup \set{(\mathsf{f},t) : t
    \in \mathcal{T}_{p}}$, and $\mathcal{E} =
  \set{\mathsf{h},\mathsf{v},\mathsf{b}}$.
\end{definition}

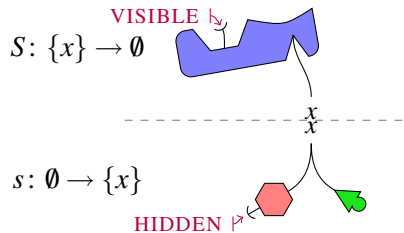
\begin{figure}[t]
  \centering
  \begin{tikzpicture}[scale=0.62]

    
    \draw (0,0) coordinate (center)
    +(up:2) coordinate (centerS)
    +(down:3) coordinate (centerS');
    
    \draw (centerS) coordinate (inactive);
    \begin{scope}[rotate around={180:(inactive)}]
      \draw (inactive)++(-2,0) coordinate (perno);
      \draw[fill=blue!50] (inactive)++(up:0.1)
      [rotate around={10:(perno)}]
      \eoneprotein{e11}{e12}{e13}
      (e11)++(down:0.5) coordinate (close1);
    \end{scope}
    \draw (centerS)
    +(2,-2) node[name=x1,below] {\small $x$};
    \draw (x1) to[out=90,in=-90] (e12);
    \draw[ -( ] (e11) to[out=90,in=-60] (close1);
    
    \draw (e11)++(left:3.15) node {$S\colon \set{x} \rightarrow \emptyset$};
    
    
    \draw[fill=green!80!gray] ($(centerS') + (2.5,1)$)
    	 [rotate around={-120:+(0,0)}]  \adp{adp1};
     \draw[fill=red!50] ($(centerS') + (1.2,0.8)$)
    	 [rotate around={120:+(0,0)}]  \mprotein{mprot}{adp2};
    
    \draw (centerS')
    +(2,2) node[name=x2,above] {\small $x$};
    \draw (adp1) to[out=150,in=-90] (x2)
    (adp2) to[out=30,in=-90] (x2);
    \draw[-)] (mprot) to ($(mprot)!1!30:+(left:0.3)$) coordinate (close2);
    
    \draw (centerS')++(-3,1.2) node {$s\colon \emptyset \to \set{x}$};
    
    
    \draw[dashed, gray] (center)++(down:0.5)
    +(left:2) -- +(right:4);

    \draw ($(e11) + (-1.5,0.7)$) node (hidden) {\small \color{purple}\textsc{visible}};
    \draw[|->, thin, purple] (hidden) to[out=0, in=120] (close1);
    
    \draw ($(close2) + (-1.5,-0.2)$) node (visible) {\small \color{purple}\textsc{hidden}};
    \draw[|->, thin, purple] (visible) to[out=0, in=210] ($(close2) + (-0.1,-0.05)$);
    
    \end{tikzpicture}
  \caption{An example of a wrong composition.}
  \label{fig:wrongcomp}
\end{figure}

We say that a site is \emph{visible}, \emph{hidden}, \emph{bond} or
\emph{free} if it is connected to an link of type $\mathsf{v}$,
$\mathsf{h}$, $\mathsf{b}$ or $\mathsf{f}$, respectively. In
principle, hyperlinks can connect more than two ports, yielding
biological meaningless link graphs because a bond involves at most two
sites. Clearly, we want this condition to be preserved under
composition (it is always preserved by tensor).  Actually, two good
protein solutions could be composed yielding an incorrect solution, as
in the Figure~\ref{fig:wrongcomp}.

Well-formed link graphs can be defined as follows:

\begin{definition}[Protein solution]
  A \emph{protein solution} $S\colon X\to Y$ is a link graph over a
  protein signature $\mathcal{P}$, such that every link has at most
  two peers, and every link of type $\mathsf{v}$ or $\mathsf{h}$ has
  exactly one peer:
  \begin{equation}\label{eq:protsol}
    \forall l:t.\; t \in \set{\mathsf{b},\mathsf{f}} \Rightarrow
    |link^{-1}(l:t)| \in \set{1,2}
    \quad \text{and} \quad
    \forall l:t.\; t \in \set{\mathsf{h},\mathsf{v}} \Rightarrow
    |link^{-1}(l:t)| = 1
    \tag{\textsc{ProtSol}}
  \end{equation}
\end{definition}
\noindent

\begin{definition}[Protein link graphs]
  The category of \emph{protein link graphs} $\cat{Lgp}$ is the
  category of link graphs sorted on the predicate \ref{eq:protsol}.
\end{definition}

\looseness=-1
Let us now consider possible reaction rules on link graphs. As for
$\kappa$-calculus, reactions allow complexations only between two
visible sites, hence tied sites must be freed before being involved in
any protein interaction.  Using the same approach of
\cite{dl:kappa04}, protein reactions can either increase or decrease
the level of connection of protein complexes; they can introduce new
nodes (protein synthesis) or remove nodes (protein
degradation). Formally, we define a \emph{growing relation}:
\begin{definition}[Growing relation]
  The \emph{growing relation $\grow$} over pairs $(\omega, S)$ 
  of wirings $\omega \colon Z \to Y$ and discrete solutions $S\colon X \to Z$ 
  is defined inductively as:
  \begin{gather*}
    \rl{ id_{\emptyset}, id_{\emptyset} \grow id_{\emptyset}, id_{\emptyset} }{  }
    \;
    \textsf{(empty)}
    \qquad\qquad
    \rl{\omega \otimes \mu, S \otimes T \grow \omega' \otimes \mu, S' \otimes T}
    { \omega, S \grow \omega', S' } \;\textsf{(tens)} 
    \\[1ex]
    \rl{/z{:\,}\mathsf{v} \otimes \omega, S \grow /z{:\,}\mathsf{h} \otimes
      \omega, S}
    {  } \;
    \textsf{(hide)}
    \qquad\quad
    \rl{ /z{:\,}\mathsf{h} \otimes \omega, S \grow /z{:\,}\mathsf{v} \otimes
      \omega, S}
    {  } \;
    \textsf{(reveal)}
    \\[1ex]
    \rl{ /z_{1}{:\,}\mathsf{v} \otimes /z_{2}{:\,}\mathsf{v} \otimes \omega,
      S \grow  (/z{:\,}\mathsf{h} \circ z/\set{z_{1},z_{2}}) \otimes \omega, S}
    {  } \;
    \textsf{(tie)}
    \\[1ex]
    \rl{/\orivec z{:\,}\mathsf{v} \otimes \omega, S \grow
      (/\orivec z{:\,}\mathsf{h}\circ
      (\orivec z/ \orivec w \otimes id_{\orivec z})) 
      \otimes \omega \otimes \mu, S \otimes T}
    {T\colon \emptyset \to \set{\orivec w} \uplus W \qquad 
      \mu\colon W \to \emptyset}
    \;
    \textsf{(synth)}
  \end{gather*}
  The growing relation can be lifted to morphisms of $\cat{Lgp}$: $G
  \grow H$ iff $\omega, G' \grow \mu, H'$, where $\omega \circ G'$,
  $\mu \circ H'$ are the discrete decompositions of $G$ and $H$,
  respectively.
\end{definition}
Let $G \grow H$, according to \textsf{(hide)} and \textsf{(reveal)}
$H$ can toggle free sites in $G$ from visible to hidden and vice
versa, whereas \textsf{(tie)} allows $H$ to bind only pairs of sites
that are visible in $G$. The \textsf{(synth)} axiom adds new proteins
in the solution $H$ (possibly connected to proteins already in $G$) if
their ports are all linked to edges in $H$.

Note that the discrete decomposition of a protein link graph
induces a separation among protein interfaces and protein nodes.  An
example is given in Fig.~\ref{fig:growex}.
\begin{figure}[t]
  \centering
  \begin{tikzpicture}[scale=0.6, font=\small]

    
    \draw (0,0)
    +(-4.25,0) coordinate (centerA)
    +(4.5,0) coordinate (centerB);

    
    
    \draw (centerA)++(up:3)
    	+(left:3.5) coordinate (upperleft)
	+(right:3) coordinate (upperright);
    

      \draw[fill=blue!60,scale=0.8] ($(centerA) + (2,1.5)$)
      [rotate around={-90:(centerA)}]
      \eoneprotein{e11a}{e12a}{e13a};
       
      \draw[fill=green!70,scale=0.8] ($(centerA) + (2,-2.25)$)
      [rotate around={-90:(centerA)}]
      \etwoprotein{e21a}{e22a};
      
      \draw
    	($(upperleft) + (down:2)$) coordinate (upleft) 
    	($(upperright) + (down:2)$)coordinate (upright);
      
      \draw ($(upleft)!(e12a)!(upright) + (right:1)$) node (z5) {$z_{5}$}
      (e12a) to[out=-10, in=-90] (z5);
      \draw ($(upleft)!(e11a)!(upright) + (left:1.5)$) node (z3) {$z_{3}$}
      (e11a) to[out=180, in=-90] (z3);
      \draw ($(upleft)!(e13a)!(upright) + (left:0.25)$) node (z4) {$z_{4}$}
      (e13a) to[out=175, in=-90] (z4);
      
      \draw ($(upleft)!(e21a)!(upright) + (right:1.5)$) node (z2) {$z_{2}$}
      (e21a) to[out=-10, in=-90] (z2);
      \draw ($(upleft)!(e22a)!(upright) + (right:0.25)$) node (z1) {$z_{1}$}
      (e22a) to[out=5, in=-90] (z1);
      
      \draw ($(z5) + (right:0.75)$) node (z6) {$z_{6}$}
      	(z6) to[out=-90,in=90] ($(z6) + (down:3.5)+(left:0.5)$) node[below] (x) {$x$};

    
    \draw[-|] ($(upperleft)!(z5)!(upperright) + (left:0.75)$) node[above] (y1) {$y_{1}$} 
    	to ($(y1) + (down:0.75)$);
    \draw (z5) to ($(upperleft)!(z5)!(upperright)$) node[above] {$y_{2}$};
    \draw (z6) to ($(upperleft)!(z6)!(upperright)$) node[above] {$y_{3}$};
    \draw (z1) .. controls +(up:2.25) and +(up:2.25) .. (z4);
    \draw[-)] (z2) to ($(z2) + (up:1)$);
    \draw[-(] (z3) to ($(z3) + (up:1)$);
    
    
    \draw ($(upperleft)!0.5!(upleft)$) node[left=0.2cm] {$\omega$};
    \draw[decorate,decoration=brace]  ($(upleft) + (up:0.1)$) -- (upperleft);
    \draw ($(upperleft)!(x)!(upleft)$) coordinate (downleft)
    	($(upleft)!0.5! (downleft)$) node[left=0.2cm] {$S$};
    \draw[decorate,decoration=brace] (downleft) -- ($(upleft) + (down:0.1)$);

    
    
    \draw (centerB)++(up:3)
    	+(left:3) coordinate (upperleft)
	+(right:5.5) coordinate (upperright);
    

      \draw[fill=blue!60,scale=0.8] ($(centerB) + (2,3.25)$)
      [rotate around={-90:(centerB)}]
      \eoneprotein{e11a}{e12a}{e13a};
       
      \draw[fill=green!70,scale=0.8] ($(centerB) + (2,-2.25)$)
      [rotate around={-90:(centerB)}]
      \etwoprotein{e21a}{e22a};
      
      \draw[fill=red!60,scale=0.8]
      [rotate around={90:($(e11a)!0.5!(e21a)$)}] 
      ($(e11a)!0.5!(e21a)$) \mprotein{m1}{m2};
      
      \draw
    	($(upperleft) + (down:2)$) coordinate (upleft) 
    	($(upperright) + (down:2)$)coordinate (upright);
      
      \draw ($(upleft)!(e12a)!(upright) + (right:1)$) node (z5) {$z_{5}$}
      (e12a) to[out=-10, in=-90] (z5);
      \draw ($(upleft)!(e11a)!(upright) + (left:1.5)$) node (z3) {$z_{3}$}
      (e11a) to[out=180, in=-90] (z3);
      \draw ($(upleft)!(e13a)!(upright) + (left:0.25)$) node (z4) {$z_{4}$}
      (e13a) to[out=175, in=-90] (z4);
      
      \draw ($(upleft)!(e21a)!(upright) + (right:1.5)$) node (z2) {$z_{2}$}
      (e21a) to[out=-10, in=-90] (z2);
      \draw ($(upleft)!(e22a)!(upright) + (right:0.25)$) node (z1) {$z_{1}$}
      (e22a) to[out=5, in=-90] (z1);
      
      \draw ($(z2) + (right:0.75)$) node (w1) {$w_{1}$}
      (m1) to[out=180, in=-90] (w1);
       \draw ($(z3) + (left:0.75)$) node (w2) {$w_{2}$}
      (m2) to[out=0, in=-90] (w2);
      
      \draw ($(z5) + (right:0.75)$) node (z6) {$z_{6}$}
      	(z6) to[out=-90,in=90] ($(z6) + (down:3.5)+(left:0.5)$) node[below] (x) {$x$};
    
    
    \draw[-|] ($(upperleft)!(z5)!(upperright) + (left:0.75)$) node[above] (y1) {$y_{1}$} 
    	to ($(y1) + (down:0.75)$);
    \draw (z5) to ($(upperleft)!(z5)!(upperright)$) node[above] {$y_{2}$};
    \draw (z6) to ($(upperleft)!(z6)!(upperright)$) node[above] {$y_{3}$};
    \draw (z1) .. controls +(up:2.25) and +(up:2.25) .. (z4);
    \draw[-(] (z2) to ($(z2) + (up:1)$);
    \draw (z3) .. controls +(up:1) and +(up:1) .. (w2);
    \draw[-(] (w1) to ($(w1) + (up:1)$);
    
    \draw ($(center) + (right:0.5)$) node {\Huge $\grow$};
    
    
    \draw ($(upperright)!0.5!(upright)$) node[right=0.2cm] {$\mu$};
    \draw[decorate, decoration=brace] (upperright) -- ($(upright) + (up:0.1)$);
    \draw ($(upperright)!(x)!(upright)$) coordinate (downright)
    	($(upright)!0.5! (downright)$) node[right=0.2cm] {$T$};
    \draw[decorate, decoration=brace] ($(upright) + (down:0.1)$) -- (downright);
    
    \end{tikzpicture}
  \caption{An example of a growing solution: $\omega, S \grow \mu, T$.}
  \label{fig:growex}
\end{figure}
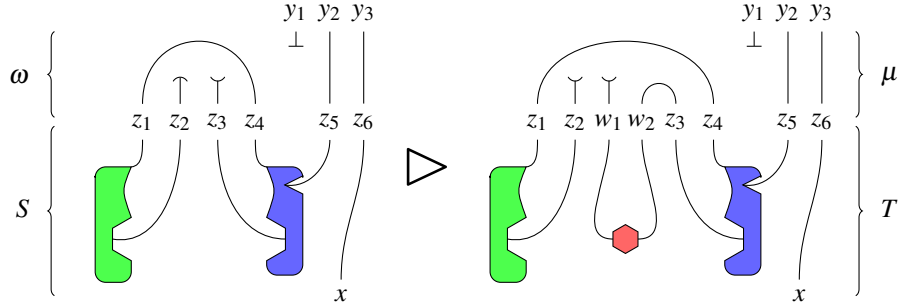

Now we can introduce the notion of \emph{protein reactive systems},
where rewriting rules respect causality. This is guaranteed by
requiring rules to either grow or shrink solutions, representing
complexations and decomplexations respectively.
\begin{definition}[Protein reactive systems]
  A rule $(l,r)$ on protein link graphs is \emph{monotone} if $l \grow
  r$ and $r$ is connected; $(r,l)$ instead is \emph{antimonotone}.

  A \emph{protein reaction system} (PRS) is a reactive system over
  $\cat{Lgp}$, whose all reaction rules are monotone or antimonotone.
\end{definition}


We establish now a formal correspondence between protein reactive
systems and protein transition systems. We use the set of
$\kappa$-protein names $\mathcal{P}$ as the set of protein controls in
$\cat{Lgp}$, using the corresponding $\kappa$-arity function $s \colon
\mathcal{P} \to \mathbb{N}$.

Now, given a set $X$ of typed names, the encoding function
$\dsb{\cdot}_X$ mapping $\kappa$-solutions into $\cat{Lgp}$ is
\begin{gather*}
  \begin{aligned}
    \dsb{\mathbf{0}}_X & = X &
    \dsb{(x)(S)}_X & =
    /y{:\,}\mathsf{b} \circ
    \dsb{S\set{y:\mathsf{b}/x:\mathsf{t}}}_{X\uplus\set{y:b}}
    \quad \text{where } t\in \{\mathsf{f},\mathsf{b}\}
    \\
    \dsb{S\kcomp T}_X & = \dsb{S}_X \parallel \dsb{T}_X &
    \dsb{A(\rho)}_X & = X \parallel
    (
    ( \textstyle
    \bigparallel_{i=1}^{s(A)} \dsb{\rho(i)}^{z_{i}}_{X})\circ
    A_{z_{1}:\mathsf{f}\dots z_{s(A)}:\mathsf{f}}
    )
  \end{aligned} \\
  \text{where} \quad 
  \dsb{h}_{X}^{z} = /z{:\,}\mathsf{h} \qquad
  \dsb{v}_{X}^{z} = /z{:\,}\mathsf{v} \qquad
  \dsb{x}_{X}^{z} = x:\mathsf{b} / z:\mathsf{f} \text{ (if $x \in X$)}
\end{gather*}
Note that $\dsb{S}_X$ is defined only when $fn(S) \subseteq X$.
Fig.~\ref{fig:enck} shows an example.
\begin{figure}[t]
  \centering
  \begin{tikzpicture}[]\input{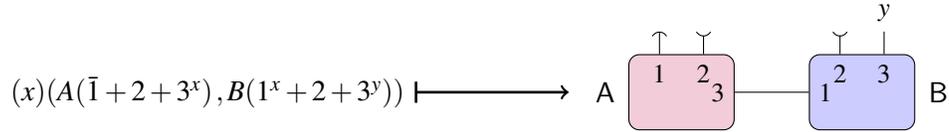}\end{tikzpicture}
  \caption{Encoding a $\kappa$-solution in a protein link graph.}
  \label{fig:enck}
\end{figure}

We can prove the following correspondence results.
\begin{proposition}[Syntax] \label{prop:kappasyntax}
  For any set of names $X$, $S\equiv T$ iff $\dsb{S}_X = \dsb{T}_X$.

  Moreover, for all $X$, the map $\dsb{\cdot}_X$ is surjective over
  the homset $\cat{Lgp}(\emptyset,X)$.
\end{proposition}

\begin{definition}[$\kappa$-PRS]
  A \emph{$\kappa$-PRS} is a reactive system over $\cat{Lgp}$ whose
  set of rules is defined as  
  \mbox{
  $\mathcal{R} = \set{(\dsb{l}_{fn(l)},\dsb{r}_{fn(r)}) \mid (l,r)\in \mathcal{R}_\kappa}$
  }, where $\mathcal{R}_\kappa$ is a $\kappa$-calculus protein transition system.
\end{definition}

\begin{proposition} \label{prop:kprs}
  Every $\kappa$-PRS is a protein reactive system.
\end{proposition}
\begin{proof}
  (sketch) Just prove that growing relations on $\kappa$-solutions and
  protein link graphs correspond.
\end{proof}

\begin{theorem}[Semantics]\label{prop:corrkappa}
  $L\leadsto (\tilde{x})R$ iff $\dsb{L}_X\rightarrowtriangle
  \dsb{(\tilde{x})R}_X$, where $\tilde{x}\cap X =\emptyset$.
\end{theorem}
\begin{proof}
  Follows from Proposition~\ref{prop:kappasyntax} and \ref{prop:kprs}. 
\end{proof}

\section{Biological bigraphs}\label{sec:biobig}

\subsection{Biobigraphs}
A cellular membrane is a closed surface that acts as a container
for biochemical solutions. One of the most remarkable properties of biological
membranes is that they form a two-dimensional fluid (a lipid bi-layer) in which
hydrophobic proteins (often only their hydrophobic sub-unit) can be immersed 
and freely diffuse. As a consequence membranes are localities where protein 
nodes can reside in them. We will represent them as two nested nodes, 
$\mcis$ and $\mext$ (for \emph{cytosolic} and \emph{extra-cytosolic} 
layers), so that membranes behaves as hydrophobic protein containers.
Formally, the cellular membrane signature is
\begin{equation*}
  \mathcal{C} \defeq \set{ \mcis \colon 0 , \mext \colon 0 }
\end{equation*}
where both controls are active, since biological reaction can happen
inside a cell.

Membranes come in different kinds, distinguished mostly by the
\emph{trans-mem\-bra\-ne} proteins embedded in them. Such proteins
have a consistent orientation and can act on both sides of the
membrane simultaneously, propagating external signals inside the
membrane-compartment (and vice versa).  Membrane behaviour is
determined by proteins embedded in them, hence membranes inherit the
consistent orientation of their trans-membrane proteins.  In our model
we express trans-membrane proteins as complexes formed by an
hydrophobic unit linked to one or more hydrophilic units (see
Fig.~\ref{fig:memb}(b)), in order to make explicit the protein
orientation and, at the same time, to formally explain the orientation
of membranes at the level of proteins
\cite{cardelli05:amsb,cardelli08:tcs}.

Membranes themselves are impermeable containers and they do not allow
large molecules (such as proteins) to simply traverse them. At the
same time protein (de-)complexation cannot take place through the
membrane, because complexations take place only if their protein
domains are close to each other.  As a consequence, protein complexes
like the one in Fig.~\ref{fig:memb}(a) are not biologically plausible.
In order to get a bigraphical category of only well-formed biological
bigraphs, we rule out these morphisms by applying the following
condition:
\begin{equation}
  \text{if } v,w \text{ peer and }
  	prnt(v) = prnt^{k}(w) 
  	\text{, then }
	\neg (\exists\, 1 \leq i < j < k. \;
        ctrl(prnt^{i}(v)), ctrl(prnt^{j}(w)) \in \mathcal{C})
  \tag{\textsc{Imp}} \label{eq:imp}
\end{equation}
which says that two nodes cannot be linked if there are more than two
occurrences of a membrane control along the path to the common
ancestor.
\begin{figure}[t]
  \centering
    \begin{tikzpicture}[scale=0.65]\input{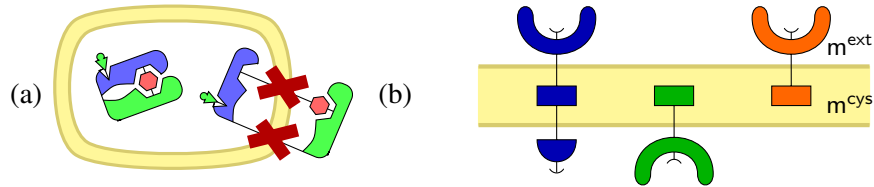}\end{tikzpicture}
    \caption{(a) Solution not respecting membrane impermeability;
      (b) Trans-membrane proteins with different orientations.}
  \label{fig:memb}
\end{figure}

This is not the only well-formedness condition to be satisfied
by a biological bigraph. In fact, with respect to the bi-layered membrane
representation, we have to rule out other ill-formed bigraphs.
\emph{A priori} one could insert $\mext$ directly into
another $\mext$, obtaining an ill-formed membrane. 
Moreover we want hydrophilic (polar) proteins not to reside on 
membranes and, analogously, hydrophobic (apolar) proteins not to
reside in water. Hence, we have to rule out systems violating 
the polarity constraint, by means of a sorting over the place graph.  
To this end, we assume that the protein signature is partitioned in 
\emph{polar} $\mathcal{P}$ and \emph{apolar} $\mathcal{A}$ controls; 
moreover, we say that $\mext$ is polar and $\mcis$ is apolar. 
Then, the polarity constraint can be expressed by the two following 
predicates (controls $\pcarry,\pmemb$ will be introduced and explained later):
\begin{align}
  \text{if } ctrl(v) \text{ is polar, then } prnt(v)
  \text{ is apolar, or } ctrl(prnt(v)) = \pcarry
  \tag{\textsc{Polar}} \label{eq:polar} \\
  \text{if } ctrl(v) \text{ is apolar, then } prnt(v)
  \text{ is polar, or } ctrl(prnt(v)) = \pmemb
  \tag{\textsc{Apolar}} \label{eq:apolar}
\end{align}

However the polarity condition does not forbid wrong systems as
those shown in Fig.~\ref{fig:bignes}(a) and (c).
\begin{figure}[t]
  \centering
  \begin{tikzpicture}[scale=0.6,font=\scriptsize]\input{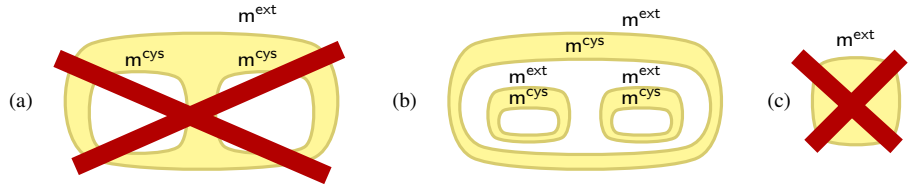}\end{tikzpicture}
  \caption{(a) Wrong nesting. (b) Good nesting. (c) Micelle.}
  \label{fig:bignes}
\end{figure}
To this end, we need also the following condition to ensure the 
correct nesting of membrane layers:
\begin{equation}
  \text{if } ctrl(v) = \mext \text{ then } \exists! v' \text{ s.t.~}
  ctrl(v') = \mcis \wedge prnt(v') = v \tag{\textsc{2Layer}} \label{eq:bi-layer}
\end{equation}
The vice versa is not necessary, indeed the system in
Fig.~\ref{fig:bignes}(b) is well formed.

Let us now consider the dynamics of biological systems.  As expected,
we recover all protein reaction rules from PRSs but extended with
locations, and requiring that proteins cannot change their position.
Moreover, taking into consideration membranes, our model must deal
also with membrane-transport.  Membrane reconfigurations are fully
justified at the protein-level; for instance, extension of pseudopods
during a phagocytosis is triggered by the reorganization of the actin
cytoskeleton \cite{au:phago}.  In our model the protein justification
for membrane transport is ensured by splitting membrane
reconfigurations into two steps:
\begin{enumerate}
\item The first step adds some special \emph{mobility controls}, which
  describe the \emph{intent} of executing an action
  (e.g.~phagocytosis) and its \emph{direction} (e.g.~what is eating
  what). Moreover, these controls ``freeze'' the sub-systems, until
  the action will be completed. 
\item The second step performs effectively the action, operating on
  the place graph, changing the position of systems (e.g., one is
  moved inside the other), removing the mobility controls and
  adding or removing a double layer.  
\end{enumerate}
In order to implement this scenario, we need two new kind of rules:
\emph{(mobility) introduction rules}, for the first step, and
\emph{(mobility) commitment rules,} for the second step, i.e., the
membrane reconfiguration.

Let us discuss the latter first. Actually, all membrane interactions
can be reduced to \emph{pinching} and \emph{fusing} of membranes
\cite{cardelli08:tcs}.
Hence, we add exactly two
mobility commitment rules, namely $pinch$ and $fuse$,
depicted in Fig.~\ref{fig:pfrules}.
\begin{figure}[t]
  \centering
  \begin{tikzpicture}[scale=0.7]\input{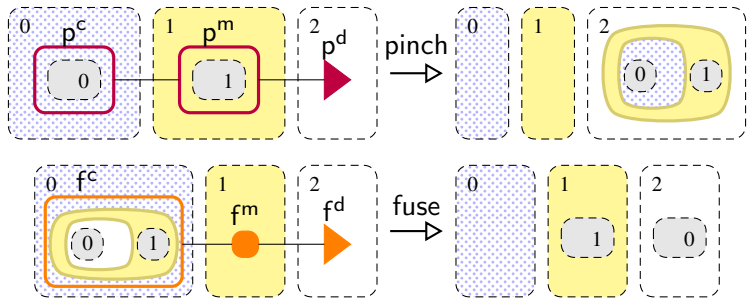}\end{tikzpicture}
  \caption{Mobility commitment rules: \textsf{pinch} and
    \textsf{fuse}.  These rules preserve \emph{bitonality} (regions with same pattern
    have the same tonality).}
  \label{fig:pfrules}
\end{figure}
These rules use some special controls to describe which parts of the
system should be moved from a location to another; these controls are
defined in the following \emph{mobility signature} 
\begin{equation*}
  \mathcal{M}\defeq \mathcal{M}_{P} \cup \mathcal{M}_{F}
  \quad \text{where} \quad
  \mathcal{M}_{P} \defeq \set{\pcarry\colon 1, \pmemb\colon 2, \pdir\colon 1} 
  \qquad
  \mathcal{M}_{F} \defeq \set{\fcell\colon 1, \fmemb\colon 2, \fdir\colon 1}
\end{equation*}
All mobility controls are polar, except for $\pmemb$ and $\fmemb$;
moreover $\pdir,\fmemb,\fdir$ are atomic whereas
$\pcarry,\pmemb,\fcell$ are passive, so that they temporally
``freeze'' the part of system that will be moved.  Notice that the
rules \textsf{pinch} and \textsf{fuse} are \emph{bitonal}, in the
sense of \cite{cardelli04:bc, cardelli08:tcs}.  Preservation of
bitonality means that reactions must preserve the even/odd parity with
which components are nested inside membranes (see
Figure~\ref{fig:pfrules}).

Mobility controls just help a mobility reaction to be modeled, and
there is not a biological reason to admit a mobility node to be peered
with a protein one.  Moreover, mobility controls make sense only if
tripled-up as in Figure~\ref{fig:pfrules}, hence we allow only these
linkage combinations.  In addition, mobility controls cannot be placed
freely. For example, there is no biological motivation to put
$\pcarry$ and $\pdir$ in the same location, since no cell can ingest
itself; similar observations are valid for fusions.  Hence, in order
to discipline these rules, we introduce a sorting to forbid illegal
membrane reconfigurations:
\begin{align}
&\begin{array}{l@{}l}
\text{if } link((v, i)) &{} = link((w,j)) \text{ then }
(ctrl(v), ctrl(w) \in \mathcal{P}\cup\mathcal{A}) \vee {} \\
& \Big( ctrl(v) = \pmemb \wedge \big( 
            (ctrl(w) = \pcarry \wedge i=0) \vee 
            (ctrl(v) = \pdir \wedge i=1) \big) \Big) \vee {} \\
& \Big( ctrl(v) = \fmemb \wedge \big( 
            (ctrl(w) = \fcell \wedge i=0) \vee 
            (ctrl(v) = \fdir \wedge i=1) \big) \Big)
\end{array} \notag \\
&\text{and if } 
ctrl(v) \in \mathcal{M} \text{ then } \forall i \in ar(ctrl(v)). \;
	\neg\exists e. \; link^{-1}(e) = \set{(v,i)}
\tag{\textsc{Mobil}} \label{eq:mob} 
\\[0.5ex]
&\begin{aligned}
\text{if }
  (ctrl(v)=\pcarry, ctrl(w)=\pdir) \vee {} &(ctrl(v)=\fcell, ctrl(w)=\fdir) \\
  & \text{ and $u$ peered with $v$ and $w$, then } prnt(v) \neq prnt(w)
\end{aligned}
\tag{\textsc{Dir}} \label{eq:dir} \\[0.5ex]
&\text{if } ctrl(v) \in \mathcal{M} \text{ then } \forall
  k>0. \; ctrl(prnt^k(v)) \notin \set{\pcarry,\pmemb,\fcell} 
\tag{\textsc{NoNesting}} \label{eq:nonest} \\[0.5ex]
&\text{if } prnt(v) = prnt(w) \text{ and } ctrl(prnt(v)) = \fcell \text{ then }
  v = w \text{ and } ctrl(v) = \mext
\tag{\textsc{Fuse}} \label{eq:fuse}
\end{align}
Intuitively, \ref{eq:mob} forbids that a mobility node peers with a
protein node, forcing it to be linked with another mobility node of
the same type ($P$ or $F$).  \ref{eq:dir} disallows mobility controls
$\pcarry$, $\pdir$ and $\fcell$, $\fdir$ to reside in the same
location.  Note that \ref{eq:imp} forces mobility controls triples to
do not violate the bi-layer impermeability, consequently such triples
are connected in a well-formed manner.  \ref{eq:nonest} guarantees
that mobility controls cannot be nested, hence disallows ambiguities
on how performs actions.  Finally, \ref{eq:fuse} forces the control
$\fcell$ to have only one child of type $\mext$, because fusion can
take place only by two cells.

We can now define \emph{biological bigraphs} as the bigraphs which
satisfy all the sorting conditions we have given so far.
\begin{definition}[Biological bigraphs]\label{def:biobg}
  Given two sets $\mathcal{P}, \mathcal{A}$ of polar and apolar
  proteins, a \emph{biological signature} is an extension of a protein
  signature and it is defined as $\biosig(\mathcal{P},\mathcal{A})
  \defeq \langle
  \mathcal{C}\cup\mathcal{P}\cup\mathcal{A}\cup\mathcal{M}, ar,
  \set{\textsf{h},\textsf{v}} \rangle$ with the condition that all
  controls in $\mathcal{P}\cup\mathcal{A}$ are atomic.

  The category of \emph{biological bigraphs,} denoted by
  $\cat{BioBg}(\biosig(\mathcal{P},\mathcal{A}))$, is the category of
  bigraphs over $\biosig(\mathcal{P},\mathcal{A})$ sorted by
  the following predicate:
  \begin{equation*}
    \ref{eq:protsol} \wedge \ref{eq:polar} \wedge \ref{eq:apolar}
    \wedge \ref{eq:bi-layer} \wedge \ref{eq:imp} \wedge \ref{eq:mob}
    \wedge \ref{eq:dir} \wedge \ref{eq:nonest}\, .
  \end{equation*}
\end{definition}

Now, let us consider the mobility introduction rules.  Even if
membrane reconfigurations are of only two kinds (pinching and fusing),
they can be triggered by distinct of protein interactions (indeed no a
single kind of membrane internalization exists, e.g.~pinocytosis and
phagocytosis).  Hence, how mobility controls are introduced is left to
the modeler (possibly a biologist), so that he can distinguish between
the variety of signaling events that cause a membrane
transformation. However, not all introductions of mobility controls
are allowed: we want introduction rules not to perform protein
reactions. Formally:

\begin{definition}[Introduction rules]
  A \emph{mobility introduction rule} is a reaction rule $(L,R)$ over 
  biobigraphs such that
  \begin{itemize}
  \item $ctrl(V_L) \cap \mathcal{M} = \emptyset$ and $ ctrl(V_L) \cap
    (\mathcal{P}\cup\mathcal{A}) \neq \emptyset$;
  \item $L = (L_0' \parallel L_1' \parallel L_2') \circ
    (L_0'' \parallel L_1'' \parallel L_2'')$,
    where $L_i',L_i''$ of width 1 for $i\in \set{0,1,2}$, \\
    $R = (L_0' \parallel L_1' \parallel L_2') \circ ( (/x,y{:
      \,}\mathsf{h} \circ (\pcarry_x \parallel (\pmemb_{xy} \mid
    id_1) \parallel (\pdir_y \mid id_1))) \otimes id_{Z}) \circ
    (L_0'' \parallel L_1'' \parallel L_2'')$, and \\
    every $z\in Z$ is the target link of two distinct ports or of a
    port of a node in $L_2''$.
  \item $L = (L_0' \parallel L_1' \parallel L_2') \circ
    (L_0'' \parallel L_1'' \parallel L_2'')$,
    where $L_i',L_i''$ of width 1 for $i\in \set{0,1,2}$ and \\
    $R = (L_0' \parallel L_1' \parallel L_2') \circ ( (/x,y{:
      \,}\mathsf{h} \circ (\fcell_x \parallel (\fmemb_{xy} \mid
    id_1) \parallel (\fdir_y \mid id_1))) \otimes id_{Z}) \circ
    (L_0'' \parallel L_1'' \parallel L_2'')$.
  \end{itemize}
\end{definition}

The first condition states that a membrane reconfiguration can be
justified only by protein interactions and not by the presence of
mobility controls. In both other conditions, we require that the left
hand side of the rule can be split into two parts of width
3. $(L_0'' \parallel L_1'' \parallel L_2'')$ specifies the proteins
that effectively takes part in the mobility action, whilst
$(L_0' \parallel L_1' \parallel L_2')$ describes any protein not
involved in the action but required to trigger it.  The extra
condition for pinch introduction guarantees that the
related commitment rule always takes place, i.e., it does not stuck
for not violating the impermeability condition.

We ensure that no protein reactions can be performed, since the redex
$L$ and reactum $R$ only differs from \emph{an} occurrence of a
mobility triple.
In Section~\ref{sec:examples} we will give two examples of how
mobility controls and introduction rules are used to model membrane
transport.

Now, we can define the notion of biological reactive systems.
\begin{definition}[BioRS]\label{def:biors}
  A \emph{biological reactive system} (BioRS) is a reactive system
  over the category of biobigraphs
  $\cat{BioBg}(\mathcal{K}_B(\mathcal{P},\mathcal{AP}))$ and equipped
  with rules $\mathcal{R}_B = \set{\mathsf{pinch},\mathsf{fuse}} \cup
  \mathcal{R}$, where $\mathcal{R}$ is a given set of protein and
  introduction rules.
\end{definition}

\subsection{Relating biological bigraphical reactive systems}
In this section we establish formal connections between biobigraphs
and protein link graphs and ``mobility-only'' bigraphs, by means of
two functors; this allows to focus on one biological aspect per time.

The \emph{protein functor} $\Fp\colon
\cat{BioBg}(\biosig(\mathcal{P},\mathcal{A})) \to
\cat{Lgp}(\mathcal{P}\cup \mathcal{A})$ forgets all the membrane and
mobility controls and polarity aspects of a system, ``flattening'' it
to a protein solution. Formally:
\begin{align*}
  & \Fp(\langle n,X\rangle) = X \\
  & \Fp((V,E,ctrl,edge,prnt,link)) = 
   (V', E', ctrl|_{V'}, edge|_{E'}, link|_{dom(link)\setminus Prt(V\setminus V')}) \\
  & V' = \set{v\in V \mid ctrl(v)\notin \mathcal{C}\cup \mathcal{M}} \qquad
    E' = E \setminus \set{link((v,i))\in E\mid v\in V \setminus V'}
\end{align*}

On the other hand, the \emph{mobility functor} $ \Fm\colon
\cat{BioBg}(\biosig(\mathcal{P},\mathcal{A})) \to \cat{MBg}$ deletes
all proteins present in the systems and focuses only on the structural
and mobility information. Formally, let $\cat{MBg} \defeq
\cat{BioBg}(\mathcal{K}_{B}(\emptyset, \emptyset))$ be the category of ``pure
mobility biobigraphs'', and define $\Fm$ as follows
\begin{align*}
  & \Fm(\langle n,X\rangle) = \langle n,X\rangle \\
  & \Fm((V,E,ctrl,edge,prnt,link)) = 
    (V', E', ctrl|_{V'}, edge|_{E'}, prnt|_{V'},
     link|_{dom(link)\setminus P(V\setminus V')}) \\
  & V'=\set{v\in V \mid ctrl(v)\notin \mathcal{P}\cup\mathcal{A}} \qquad
    E'=E\setminus \set{link((v,i))\in E\mid v\in V \setminus V'}
\end{align*}

These functors allows to give a formal justification of membrane-level
actions in terms of simpler interactions, as it is formalized in the
next result.
\begin{theorem}\label{prop:projections}
  Let $\mathcal{D}$ be a BioRS over
  $\cat{BioBg}(\biosig(\mathcal{P},\mathcal{A}))$, with rules
  $\mathcal{R}$.
  Let $\Fp(\mathcal{D})$ be the protein reactive system
  over $\cat{Lgp}(\mathcal{P}\cup\mathcal{A})$ whose rules are
  $\Fp(\mathcal{R}\setminus \{pinch, fuse\})$.
  Let $\Fm(\mathcal{D})$ be the BioRS over $\cat{MBg}$
  whose rules are $\Fm(\mathcal{R})$.

  For all $G,H$ ground bigraphs in
  $\cat{BioBg}(\biosig(\mathcal{P},\mathcal{A}))$, if
  $G\rightarrowtriangle H$ in $\mathcal{D}$ then
  \begin{itemize}
  \item either $\Fp(G) \neq \Fp(H)$ or
    $\Fm(G) \neq \Fm(H)$;
  \item if $\Fp(G) \neq \Fp(H)$ then
    $\Fp(G) \rightarrowtriangle \Fp(H)$ in
    $\Fp(\mathcal{D})$;
  \item if $\Fm(G) \neq \Fm(H)$ then
    $\Fm(G) \rightarrowtriangle \Fm(H)$ in
    $\Fm(\mathcal{D})$.
  \end{itemize}
\end{theorem}
Intuitively, this theorem states that the protein and mobility
functors ``project'' the execution traces of a biobigraphical system
into $\cat{Lgp}$ and $\cat{MBg}$, respectively, as shown in the
diagram in Figure~\ref{fig:diagram}. In the first case, only the
interactions between proteins are observed, but not their structural
effects; in the latter case, the evolution of mobility controls and
positions of cell is observed, but not the protein interactions which
actually cause these dynamics.

\begin{figure}[t]
  \centering
  \begin{tikzpicture}[scale=0.9]\input{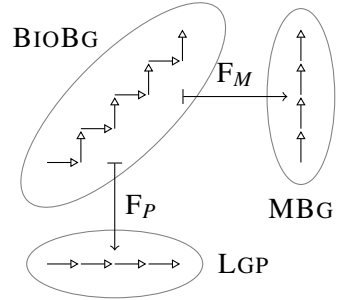}\end{tikzpicture}
  \caption{Mapping biobigraphical steps to either protein or
    mobility reactions.}
  \label{fig:diagram}
\end{figure}

\section{Examples}\label{sec:examples}

\subsection{Vesicles formation}\label{sec:vesicle}

In this section we give a formal bigraphical description of the
vesiculation process 
which is one of the most fundamental transport mechanisms in cells.

Vesicle transport is used to translocate proteins and membranes from
the endoplasmatic reticulum (ER) to the plasma membrane via the Golgi
apparatus, or to ingest macromolecules by receptor-mediated
endocytosis at the plasma membrane.  Vesicles form by budding from
membranes. Each bud has distinctive coat protein on cytosol surface,
which shapes the membrane to form a bubble. The bud captures the
correct molecules for outward transport by a selective bond with a
cargo receptor attached through the double layer to the coat protein
complex, which is lost after the budding completes. Here we consider
the endocytosis at the plasma membrane where clathrin coated vesicles
are formed.

In order to represent this vesiculation process, we introduce the
protein signature 
\[
 \mathcal{P}_v = \set{ \mathsf{cargo} \colon 1,
  \mathsf{rec}^{cys} \colon 2,\mathsf{rec}^{ext} \colon 2,
  \mathsf{adpt} \colon 2, \mathsf{clath} \colon 1 }
\quad \text{and}\quad
 \mathcal{A}_v = \set{ \mathsf{rec}^{m} \colon 2 }
\]
and finally the corresponding reaction rules for their interactions
(graphically depicted in Figure~\ref{fig:vesiclerules}):
\begin{align*}
  &
  (/x,y{:\,}\mathsf{v} \otimes /z,w,k{:\,}\mathsf{h}) \circ 
  ((\mathsf{cargo}_{x} \mid \mathsf{rec}^{ext}_{yz}) \parallel
  \mathsf{rec}^{m}_{zw} \parallel \mathsf{rec}^{cys}_{wk})
  \xrightarrowtriangle{} \tag{\textsf{rec}} \\
  & \phantom{(/x,y{:\,}\mathsf{v} \otimes /z,w,k{:\,}\mathsf{h}) \circ ((}
  (/k{:\,}\mathsf{v} \otimes /b, z,w{:\,}\mathsf{h}) \circ
  (\mathsf{cargo}_{b} \mid \mathsf{rec}^{ext}_{bz}) \parallel
  \mathsf{rec}^{m}_{zw} \parallel \mathsf{cys}^{cys}_{wk})
  \\
  &
  (/y{:\,}\mathsf{h}\otimes /x,k{:\,}\mathsf{v}) \circ
  (\mathsf{rec}^{cys}_{wk} \mid \mathsf{adpt}_{xy}) 
  \xrightarrowtriangle{}  \tag{\textsf{adpt}}
  (/y{:\,}\mathsf{v} \otimes /b{:\,}\mathsf{h}) \circ
  (\mathsf{rec}^{cys}_{wb} \mid \mathsf{adpt}_{by})
  \\
  &
  /yz{:\,}\mathsf{v} \circ (\mathsf{adpt}_{xy} \mid \mathsf{clath}_{z}) 
  \xrightarrowtriangle{} \tag{\textsf{coat}}
  /b{:\,}\mathsf{h} \circ (\mathsf{adpt}_{xb} \mid \mathsf{clath}_{b})
  \\
  &
  /b,b'{:\,}\mathsf{h} {\circ} (\mathsf{rec}^{cys}_{wb} \mid
  \mathsf{adpt}_{bb'} \mid \mathsf{clath}_{b'}) 
  \xrightarrowtriangle{} \tag{\textsf{uncoat}}
  (/x,k{:\,}\mathsf{v} \otimes /b{:\,}\mathsf{h}) {\circ}
  (\mathsf{rec}^{cys}_{wk} \mid \mathsf{adpt}_{xb'} \mid \mathsf{clath}_{b'})
\end{align*}

\begin{figure}
\centering
  \def\width{2.3}
  \pgfdeclarelayer{background} 
  \pgfsetlayers{background,main}
  \begin{tikzpicture}[
  	memb/.style={thick, double distance= 0.23*\width cm, 
		draw= yellow!60!gray, double=yellow!50},
  	receptor/.style={fill=blue!65},
	cargo/.style={fill=red!95},
	clathrin/.style={fill=green!75!blue},
	adaptin/.style={fill=orange},
	pinch/.style={thick, draw=purple}
  ]\input{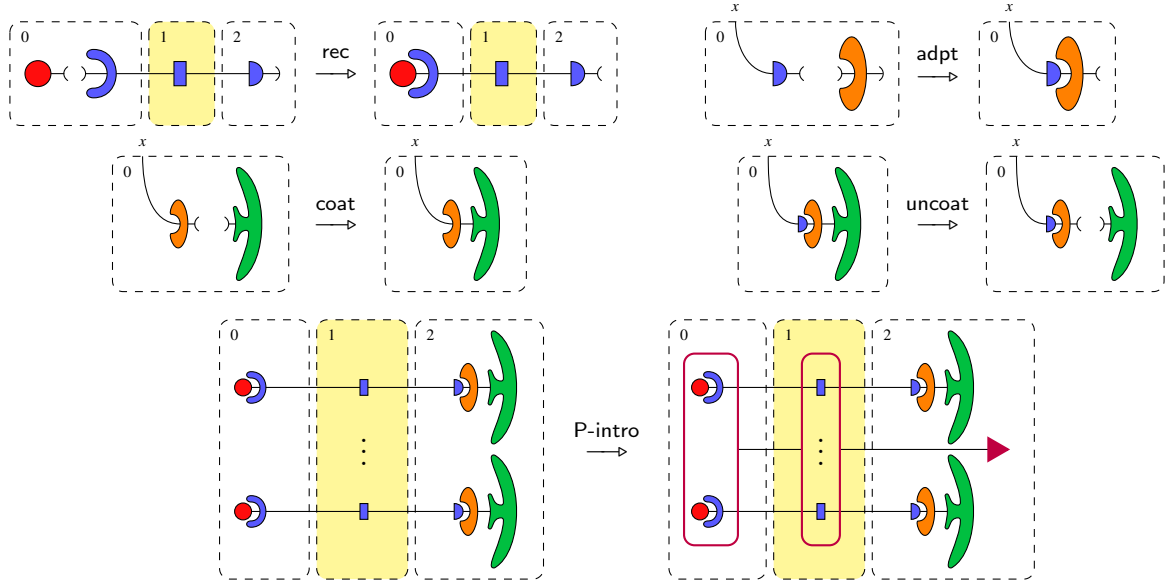}\end{tikzpicture}
  \caption{Rules for vesicles formation.}
  \label{fig:vesiclerules}
\end{figure}

The vesicle is spawn when the correct number of
clathrin-adaptin-receptor-cargo molecule complexes is formed.  Let
$(\textit{P-cmplx})^{n}$ be the bigraph representing $n$ complexes,
defined as:
\begin{align*}
  \textit{P-cmplx} & \defeq 
  /x,y,z,w,k{:\,}\mathsf{h} \circ 
  (
  (\mathsf{cargo}_{x} \mid \mathsf{rec}^{ext}_{xy} \mid id_{1}) \parallel
  (\mathsf{rec}^{m}_{yz} \mid id_{1}) \parallel
  (\mathsf{rec}^{cys}_{zw} \mid \mathsf{adpt}_{wk} \mid \mathsf{clath}_{k}
  \mid id_{1})
  )
  \\
  (\textit{P-cmplx})^0 & \defeq (1 \parallel 1\parallel 1) \qquad\qquad
  (\textit{P-cmplx})^{n+1} \defeq \textit{P-cmplx}\circ (\textit{P-cmplx})^n
\end{align*}
Then, the actual creation of membranes is formalized by a single
mobility introduction rule, which is triggered in presence of the
right number of complexes:
\begin{equation*}
  (\textit{P-cmplx})^{n} \xrightarrowtriangle{} \tag{\textsf{P-intro}}
  /x,y{:\,}\mathsf{h} \circ
  (\pcarry_{x} \parallel \pmemb_{xy} \parallel \pdir_{y}) \circ
  (\textit{P-cmplx})^{n}
\end{equation*}
After the application of this rule, the vesiculation process is
completed by the \textsf{pinch} rule.  An execution trace of this
process is graphically depicted in Fig.~\ref{fig:vesicle}.
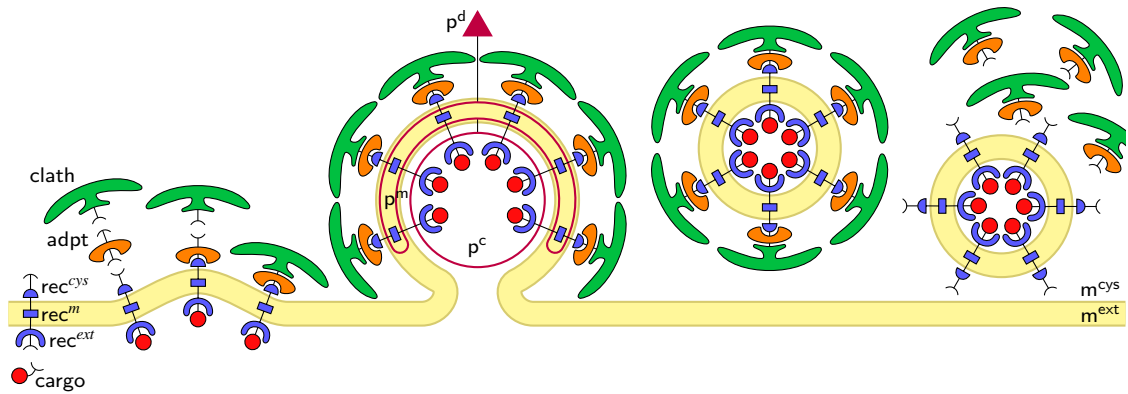
\begin{figure}[t]
  \centering
  \def\width{1.25}
  \begin{tikzpicture}[scale=1.1,
  	memb/.style={thick, double distance= 0.23*\width cm, draw= yellow!60!gray, double=yellow!50},
  	receptor/.style={fill=blue!65},
	cargo/.style={fill=red!95},
	clathrin/.style={fill=green!75!blue},
	adaptin/.style={fill=orange},
	pinch/.style={thick, draw=purple}
  ]


\clip (-7,-3) rectangle (6.5,2);
	
	\draw 
		(-7,-2) coordinate (membStart)
		($(membStart) + (right:13.5)$) coordinate (membEnd)
		($(membStart)!0.37!(membEnd)$) coordinate (budStart)
		($(membStart)!0.47!(membEnd)$) coordinate (budEnd);

	\draw[memb] ($(membStart)!0.42!(membEnd) + 1.1*(up:\width)$) coordinate (perno2)
		($(perno2)!{sin(60)*\width} cm!(budStart)$) coordinate (broken1)
		($(perno2)!{sin(60)*\width} cm!(budEnd)$) coordinate (broken2);
	\draw[memb]
		[rounded corners]
		(membStart) -- ($(membStart)!0.1!(membEnd)$) ..
		controls ($(membStart)!0.17!(membEnd) + (up:0.5)$)
		.. ($(membStart)!0.24!(membEnd)$)
			coordinate[pos=0.05] (r1)
			coordinate[pos=0.5] (r2)
			coordinate[pos=0.95] (r3)
		[sharp corners]
		-- (budStart)
		.. controls +($0.2*(right:\width)$) and ($(broken1)!{0.2*sin(60)*\width}cm! -90:(perno2)$) ..
		(broken1)
		.. controls ($(broken1)!{0.2*sin(60)*\width}cm! 90:(perno2)$) and +($0.5*sin(60)*(down:\width)$) ..
		($(perno2) + sin(60)*(left:\width)$) 
		.. controls +($0.555*sin(60)*(up:\width)$) and +($0.555*sin(60)*(left:\width)$) ..
		($(perno2) + sin(60)*(up:\width)$)
		.. controls +($0.555*sin(60)*(right:\width)$) and +($0.555*sin(60)*(up:\width)$) ..
		($(perno2) + sin(60)*(right:\width)$) 
		.. controls +($0.5*sin(60)*(down:\width)$) and ($(broken2)!{0.2*sin(60)*\width}cm! -90:(perno2)$) ..
		(broken2)
		.. controls  ($(broken2)!{0.2*sin(60)*\width}cm! 90:(perno2)$) and +($0.2*(left:\width)$) ..
		(budEnd) 
		
		-- (membEnd);
	\draw (membEnd)+(left:0.3) node {\scriptsize $\mext$};
	\draw (membEnd)+(-0.3,0.3) node {\scriptsize $\mcis$};
	
	\draw[receptor] ($(membStart)!0.02!(membEnd)$) coordinate (rec)
		\receptor{1}{0.35*\width}{recNorth}{recSouth};
	\draw ($(membStart)!0.01!(membEnd) + (down:0.75)$) coordinate (c)
		[cargo, rotate around={-90:(c)}] \cargo{0.075*\width}{cargoNorth};
	\draw[-(] (cargoNorth) to[out=0, in=-120] ($(cargoNorth) + 0.1*(up:\width) + 0.1*(right:\width)$);
	\draw[-(] (recSouth) -- +($0.15*(down:\width)$);
	\draw[-)] (recNorth) -- +($0.15*(up:\width)$);
	
	\draw (recNorth)+(up:0.05) node[right] {\scriptsize $\mathsf{rec}^{cys}$};
	\draw (recSouth)+(0.1,-0.05) node[right] {\scriptsize $\mathsf{rec}^{ext}$};
	\draw (rec) node[right] {\scriptsize $\mathsf{rec}^{m}$};
	\draw (cargoNorth)+(0.4,-0.1) node {\scriptsize $\mathsf{cargo}$};
	
	\begin{scope}[rotate around={20:(r1)}]
		\draw[receptor] (r1) \receptor{1}{0.35*\width}{recNorth}{recSouth};
		\draw[cargo] ($(recSouth)+0.15*(down:\width)$) \cargo{0.075*\width}{cargoNorth};
		\draw[adaptin] ($(recNorth) + 0.4*(up:\width)$) \adaptin{0.35*\width}{adptNorth}{adptSouth};
		\draw[clathrin] ($(adptNorth) + 0.4*(up:\width)$) \clathrin{\width}{clathSouth};
		
		\draw (recSouth) -- (cargoNorth);
		\draw[-(] (recNorth) -- +($0.15*(up:\width)$);
		\draw[-(] (adptSouth) -- +($0.15*(down:\width)$);
		\draw[-)] (adptNorth) -- +($0.13*(up:\width)$);
		\draw[-(] (clathSouth) -- +($0.15*(down:\width)$);
	\end{scope}
	\draw (adptSouth)+(-0.5,0.1) node {\scriptsize $\mathsf{adpt}$};
	\draw (clathSouth)+(-0.5,0.4) node {\scriptsize $\mathsf{clath}$};
	
	\draw[receptor] (r2) \receptor{1}{0.35*\width}{recNorth}{recSouth};
	\draw[cargo] ($(recSouth)+0.15*(down:\width)$) \cargo{0.075*\width}{cargoNorth};
	\draw[adaptin] ($(recNorth) + 0.07*(up:\width)$) \adaptin{0.35*\width}{adptNorth}{adptSouth};
	\draw[clathrin] ($(adptNorth) + 0.4*(up:\width)$) \clathrin{\width}{clathSouth};
		
	\draw (recSouth) -- (cargoNorth);
	\draw (recNorth) -- (adptSouth);
	\draw[-(] (adptNorth) -- +($0.13*(up:\width)$);
	\draw[-(] (clathSouth) -- +($0.15*(down:\width)$);
	
	\begin{scope}[rotate around={-20:(r3)}]
		\draw[receptor] (r3) \receptor{1}{0.35*\width}{recNorth}{recSouth};
		\draw[cargo] ($(recSouth)+0.15*(down:\width)$) \cargo{0.075*\width}{cargoNorth};
		\draw[adaptin] ($(recNorth) + 0.07*(up:\width)$) \adaptin{0.35*\width}{adptNorth}{adptSouth};
		\draw[clathrin] ($(adptNorth) + 0.07*(up:\width)$) \clathrin{\width}{clathSouth};
		
		\draw (recSouth) -- (cargoNorth);
		\draw (recNorth) -- (adptSouth);
		\draw (adptNorth) -- (clathSouth);
	\end{scope}
	
	\def\pmRadius{sin(60)*\width}
	\draw (perno2) -- ($(perno2) + 2.1*\pmRadius*(up:1)$);
	\draw[pinch, fill=white] (perno2) circle ({0.75*\pmRadius} cm);
	
	\draw[pinch, double distance=0.15*\width cm, double=yellow!50, cap=round]
		($(perno2)!1!-210:($(perno2) + \pmRadius*(left:1)$)$)
		arc (-30:210: {\pmRadius} cm);
	
	\def\pdLato{0.3*\pmRadius}
	\draw[pinch, fill=purple]
		($(perno2) + 2.1*\pmRadius*(up:1)$) coordinate (pdStart) --
		($ (pdStart)! {\pdLato} cm! -30: ($(pdStart) + (down:1)$) $) --
		($ (pdStart)! {\pdLato} cm! 30: ($(pdStart) + (down:1)$) $)
		-- cycle;
	\draw ($(perno2) + 0.75*\pmRadius*(down:1)$) node[above] {\scriptsize $\pcarry$};
	\draw ($(perno2) + \pmRadius*0.9*(left:1)$) node {\scriptsize $\pmemb$};
	\draw ($(perno2) + \pmRadius*2*(up:1)$) node[left] {\scriptsize $\pdir$};
	
	\foreach \angle in {-112.5, -67.5, ..., 112.5} {
		\begin{scope}[rotate around={\angle:(perno2)}]
		\draw[receptor] ($(perno2) + sin(60)*(up:\width)$) \receptor{1.5}{0.35*\width}{recNorth}{recSouth};
		\draw[cargo] ($(recSouth)+0.15*(down:\width)$) \cargo{0.075*\width}{cargoNorth};
		\draw[adaptin] ($(recNorth) + 0.07*(up:\width)$) \adaptin{0.35*\width}{adptNorth}{adptSouth};
		\draw[clathrin] ($(adptNorth) + 0.07*(up:\width)$) \clathrin{\width}{clathSouth};
		
		\draw (recSouth) -- (cargoNorth);
		\draw (recNorth) -- (adptSouth);
		\draw (adptNorth) -- (clathSouth);
		\end{scope}
	}
	
	\draw (2.2,0) coordinate (perno);
	\draw[memb] (perno) circle ({sin(60)*\width*0.66} cm);
	
	\foreach \angle in {0, 60, ..., 300} {
		\begin{scope}[rotate around={\angle:(perno)}]
		\draw[receptor] ($(perno) + sin(60)*0.66*(up:\width)$) \receptor{1}{0.35*\width}{recNorth}{recSouth};
		\draw[cargo] ($(recSouth)+0.15*(down:\width)$) \cargo{0.075*\width}{cargoNorth};
		\draw[adaptin] ($(recNorth) + 0.07*(up:\width)$) \adaptin{0.35*\width}{adptNorth}{adptSouth};
		\draw[clathrin] ($(adptNorth) + 0.07*(up:\width)$) \clathrin{\width}{clathSouth};
		
		\draw (recSouth) -- (cargoNorth);
		\draw (recNorth) -- (adptSouth);
		\draw (adptNorth) -- (clathSouth);
		\end{scope}
	}
	
	\draw (5,-0.7) coordinate (perno);
	\draw[memb] (perno) circle ({sin(60)*\width*0.66} cm);
	
	\foreach \angle in {30, 90, ..., 330} {
		\begin{scope}[rotate around={\angle:(perno)}]
		\draw[receptor] ($(perno) + sin(60)*0.66*(up:\width)$) \receptor{1}{0.35*\width}{recNorth}{recSouth};
		\draw[cargo] ($(recSouth)+0.15*(down:\width)$) \cargo{0.075*\width}{cargoNorth};
		
		\draw (recSouth) -- (cargoNorth);
		\draw[-(] (recNorth) -- +($0.15*(up:\width)$);
		\end{scope}
	}
	
	\foreach \x / \y / \angle in {-0.3/0/30, 0.2/-0.7/10, 1/-0.2/-50, 1.3/-1.3/-35} {
		\draw ($(perno) + (up:2) + (\x, \y)$) coordinate (pRnd);
		
		\begin{scope}[rotate around={\angle:(pRnd)}]
		\draw[adaptin] ($(pRnd) + 0.07*(down:\width)$) \adaptin{0.35*\width}{adptNorth}{adptSouth};
		\draw[clathrin] ($(adptNorth) + 0.07*(up:\width)$) \clathrin{\width}{clathSouth};
		
		\draw[-(] (adptSouth) -- +($0.15*(down:\width)$);
		\draw (adptNorth) -- (clathSouth);
		\end{scope}
	}\end{tikzpicture}
  \caption{Receptor-mediated endocytosis pathway (compare with \cite[Fig.~15.19]{alberts:ecb}).}
  \label{fig:vesicle}
\end{figure}

\subsection{Fc receptor-mediated phagocytosis}\label{sec:phago}

Phagocytosis is the process whereby cells engulf large particles.
Phagocytosis is triggered by the interaction of antibodies that cover
the particle to be internalized with specific receptors on the surface
of the phagocyte. Here we present a simplified example of signaling
transduction pathway for FcR-mediated phagocytosis taken from
\cite{ggc:phago}. Fc receptors bind to the Fc portion of
immunoglobulins (IgG), which trigger cross-linking of FcR and, hence,
the activation of their enzymatic sub-units. This initiates a variety
of signals which lead through the reorganization of the actin
cytoskeleton, and membrane remodeling, to the formation of the
phagosome. Here we focus mainly on the modeling criteria by which we
can encode as a BioRS a phagocytosis reaction using introduction and
commitment rules. Due to lack of space we show only the mobility
introduction rule; the design of protein rules is similar to that in
Section\ref{sec:vesicle}, and their definition can be easily recovered
by the signaling pathway shown in Fig.~\ref{fig:phago}.

\begin{figure}[t]
  \centering
   \includegraphics[width=0.85\textwidth]{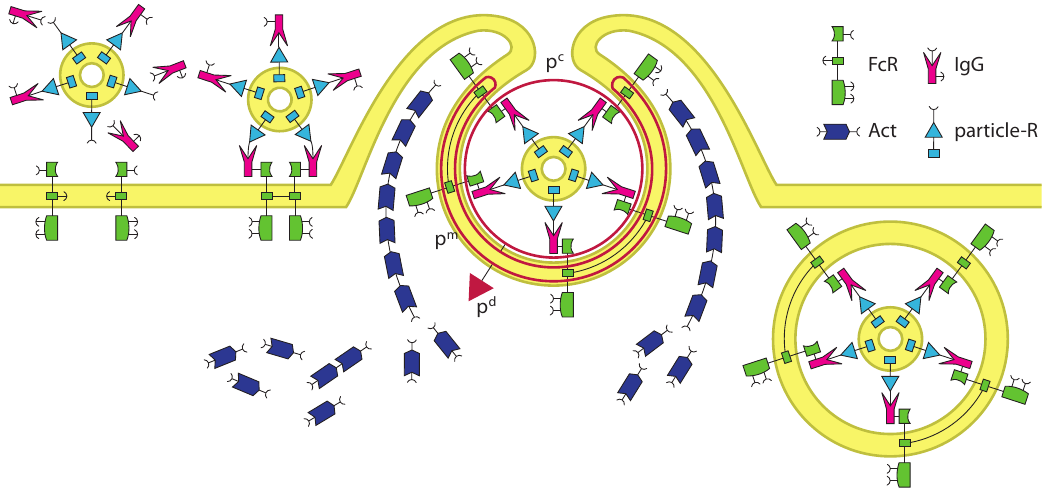}
   \includegraphics[width=0.75\textwidth]{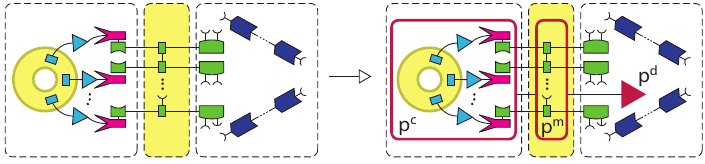}
  \caption{Fc receptor-mediated phagocytosis of an IgG-opsonized particle.}
  \label{fig:phago}
\end{figure}
 
\vspace{-1ex}

\section{Conclusions}\label{sec:concl}
\vspace{-.9ex}
In this paper we have presented a framework for protein and membrane
interactions, based on Milner's bigraphs.  Assuming the protein-level
interactions as the fundamental machine, we have characterized
formally a class of link graphs and their reactive systems which
correspond exactly to protein solutions and protein transition systems
definable in the $\kappa$-calculus \cite{dl:kappa04}.  Then we have
extended our approach to the membrane machine, whose mobility activity
can be rendered just by two fixed rules. As example applications of
this framework, we have modelled the formation of vesicles, and the Fc
receptor-mediated phagocytosis, giving a formal justification at
protein-level of membrane-level actions.


%


\paragraph{Related and future work.}
The idea of representing biological systems using bigraphs goes back
to Cardelli and Milner \cite{cardelli04:bioware,milner:ic06}. To our
knowledge our framework is the first ``foundational metamodel'',
adequate at the protein-level, and allowing to describe formally the
relations between biological models at different abstraction levels,
according to the ``tower of models'' vision.

Damgaard \emph{et al.}  have developed the $\mathcal{C}$-calculus
\cite{ddk08:langcell}, an interesting language for modelling
interactions of membranes and proteins.  There are many relevant
differences between biobigraphs and the $\mathcal{C}$-calculus.
First, the $\mathcal{C}$-calculus is focused on protein diffusion
through ``channels'', and it does not allow for phagocytosis as the
one described using biobigraphs in Section~\ref{sec:phago}.  Moreover,
one of the aim of biobigraphs is to explain higher-level interactions
with lower-level reactions, according to the ``tower of models''
vision, and to be adequate with respect to the
$\kappa$-calculus. According to this vision, in biobigraphs membrane
reconfigurations are logically diffent from protein reactions, but the
connection between the two aspects is formally specified.  In the
$\mathcal{C}$-calculus, mobility and protein reactions are more
intertwined, and a formal connection between $\mathcal{C}$-calculus
and $\kappa$-calculus has not been investigated yet.


In fact, we can say that our approach is \emph{model-driven}, that is,
we have striven for an adequate (bigraphical) model \emph{before}
defining any language.
In fact, in \cite{bgm:biobeta} we present Bio$\beta$, a new
(meta-)language inspired and corresponding to the bigraphical
framework presented in this paper, and covering both protein and
membrane aspects.  This language can be given an efficient, \emph{ad
  hoc} implementation.


A stochastic version of bigraphical reactive systems has been introduced in
\cite{kmt:mfps08}, for which systems biology is indicated as the primary
application.  Our work is complementary, and orthogonal, to
\cite{kmt:mfps08};
indeed, we can easily think of ``stochastic biobigraphs'', endowing
rules with rates, for dealing also with quantitative aspects, like
reaction speeds, concentration ratios, etc., and eventually also the
gene machine, in the line of recent approaches
\cite{cardelli07:prme,ch:psasb}.


Another interesting aspect to investigate is the bisimilarity
definable via the IPO construction on bigraphs \cite{milner:ic06},
which is always a congruence. This can be useful e.g. in
\emph{synthetic biology}, for instance, to verify if a synthetic
protein (or even a subsystem to be implanted) behaves as the natural
one.

\vspace{-0.7ex}

\bibliographystyle{abbrv}
\bibliography{allbib,pi}


\begin{extended}

\appendix

\section{BiLog: a spatial logic for bigraphs}\label{sec:bilog}
	
In this section we recall the main definitions and results about
BiLog. For a detailed survey see \cite{cms:icalp05}.

BiLog is a logic for describing the structure and sub-structure of
bigraphs. The models of this logic are formed by \emph{terms}: the
elementary terms ($\Omega$) are defined by providing a set of unary
constructors ($\Theta$), then the more complex terms are constructed
by composition ($\circ$) and tensor product ($\otimes$).
\begin{align}
  \Omega & ::= id_{I} \mid \dots
  \tag{for every $\Omega\in \Theta$}
  \\
  A, B & ::= \mathbf{F} \mid A \Rightarrow B \mid id \mid \Omega \mid
  A \otimes B \mid A \circ B \mid A \leftcompadj B \mid A \lefttensoradj B
  \nonumber
\end{align}
Terms are considered up to the following structural congruence $\equiv$:
\begin{gather*}
  G \equiv G \qquad G \equiv F \Rightarrow F \equiv G \qquad 
  G \equiv G' \wedge F \equiv F' \Rightarrow G \circ F \equiv G' \circ F'
  \\
  G \equiv F \wedge F \equiv H \Rightarrow G \equiv H \qquad 
  G \equiv G' \wedge F \equiv F' \Rightarrow G \otimes F \equiv G' \otimes F'
  \\
  G \circ id_{I} \equiv id_{J} \circ G \qquad
  (G \circ F) \circ H \equiv G \circ (F \circ H)
  \\
  id_{I} \otimes id_{j} \equiv id_{I \otimes J} \qquad
  (G \otimes F) \otimes H \equiv G \otimes (F \otimes H) \qquad
  G \otimes id_{\epsilon} \equiv G \equiv id_{\epsilon} \otimes G
  \\
  (G_{0} \otimes F_{0}) \circ (G_{1} \otimes F_{1}) \equiv (G_{0}
  \circ G_{1}) \otimes (F_{0} \circ F_{1})
\end{gather*}
	
The logic semantics is defined as follows:
\begin{alignat*}{2}
  &G \vDash \mathbf{F} &{}\iff{}
  & \text{never} \\
  &G \vDash \Omega &{}\iff{}
  &G \equiv \Omega \\
  &G \vDash id &{}\iff{}
  &\exists I. \; G \equiv id_{I} \\
  &G \vDash A \Rightarrow B &{}\iff{}
  &\text{if } G \vDash A \text{ then } G \vDash B \\
  &G \vDash A \circ B &{}\iff{}
  &\exists G', G''. \
  G \equiv G' \circ G'' \wedge G' \vDash A \wedge G'' \vDash B \\
  &G \vDash A \otimes B &{}\iff{}
  &\exists G', G''. \;
  G \equiv G' \otimes G'' \wedge G' \vDash A \wedge G'' \vDash B \\
  &G \vDash A \leftcompadj B &{}\iff{}
  &\forall G'. \;
  \text{if } G' \vDash A \text{ and } (G' \circ G)
  \text{ is defined, then } G' \circ G \vDash B \\
  &G \vDash A \lefttensoradj B &{}\iff{}
  &\forall G'. \;
  \text{if } G' \vDash A \text{ and } (G' \otimes G)
  \text{ id defined, then } G' \otimes G \vDash B
\end{alignat*}

This kind of logic assigns to each constructor in $\Theta$ a logic
constant $\Omega$ whose meaning is simply to be congruent to
the corresponding constructor.  The formula for the horizontal
decomposition $A \otimes B$ is satisfied whenever the term can be
split using a tensor product in two sub-terms whose satisfies $A$ and
$B$, respectively.  Recall that the tensor product does not allow the
sharing of anything between two terms, this fact is quite important
when we will instantiate the logic to link graphs (as we will see
later on).  The formula for the vertical decomposition $A \circ B$ is
satisfied whenever the term ca be decomposed in an upper part and a
below one (that if composed give back the original term) whose
satisfies $A$ and $B$, respectively.  The adjoints for composition
($\leftcompadj$) and tensor ($\lefttensoradj$) are considered.
	
\begin{theorem}[Logic Equivalence]
  Let $G$ and $F$ be two terms.
  \begin{enumerate}
  \item  $G \vDash A$ and $G \equiv H$ then $F \vDash A$.
  \item $G =_{L} F$ iff $G \equiv F$;
  \end{enumerate}
  where $=_{L}$ is the equivalence relation induced by the
  equisatisfiability of two terms.
\end{theorem}
	
Using the basic operators, we can define the following derived ones:
\begin{gather*}
  A_{I} \defeq A \circ id_{I} \qquad
  A_{\rightarrow J} \defeq id_{J} \circ A \qquad
  A_{I \rightarrow J} \defeq  (A_{I})_{\rightarrow J} \qquad
  A \circ_{I} B \defeq  A \circ id_{I} \circ B \\
  A \leftcompadj_{J} B \defeq A_{\rightarrow J} \leftcompadj B \qquad
  A \ominus B \defeq \neg(\neg A \otimes \neg B) \qquad
  A \bullet B \defeq \neg(\neg A \circ \neg B) \\
  A^{\exists\otimes} \defeq \mathbf{T} \otimes A \otimes \mathbf{T} \quad
  A^{\forall\otimes} \defeq \mathbf{F} \ominus A \ominus \mathbf{F} \quad
  A^{\exists\circ} \defeq \mathbf{T} \circ A \circ \mathbf{T} \quad
  A^{\forall\circ} \defeq \mathbf{F} \bullet A \bullet \mathbf{F} 
\end{gather*}

Intuitively, a link graph satisfies:
\begin{itemize}
\item $A^{\exists\otimes}$ iff it can be decomposed in some horizontal
  (i.e., tensorial) components, at least one of which satisfies $A$;
  similarly for $A^{\exists\circ}$ and vertical decomposition.
\item $A^{\forall\otimes}$ iff all possible horizontal decompositions
  have a component satisfying $A$; similarly for $A^{\forall\circ}$
  and vertical decomposition.
\end{itemize}

\smallskip
\noindent\textbf{Logic for Link Graphs} For instantiating the logic on
link graphs we consider the following set of constructor:
\begin{equation*}
  \Theta_{L} = \set{
    id_a : a \to a, \
    /a{:\,} \mathsf{t} \colon \set{a} \rightarrow \emptyset, \
    a/X \colon X \rightarrow a, \
    K_{\orivec{a}} \colon \emptyset \rightarrow \orivec{a}
  }
\end{equation*}
Given a signature $\langle \mathcal{K}, ar, \mathcal{E} \rangle$. The
constructor $/a{:\,} \mathsf{t}$ maps the name $a$ into an edge of
type $\mathsf{t}\in \mathcal{E}$.  $a/X$ maps all the names in $X$ to
$a$, it is rather interesting to note that the special case
$a/\emptyset$ (or simply $a$) represent the introduction of the name
$a$.  Let $K\in \mathcal{K}$ with arity $ar(K) = |\orivec{a}|$ we
introduce the constructor $K_{\orivec{a}} \colon \emptyset \rightarrow
\orivec{a}$, that links the ports of $K$ to the names in $\orivec{a}$.
	
Now, we have to extend the structural congruence $\equiv$ for covering
the new constructors by adding the following axioms:
\begin{gather*}
   a/a \equiv id_{a} \qquad
   /a{:\,} \mathsf{t} \circ a/b \equiv /b{:\,} \mathsf{t} \qquad
   /a{:\,} \mathsf{t} \circ a \equiv id_{\epsilon} \qquad
   \\
   b/(Y \uplus a) \circ (id_{Y} \otimes a/X) \equiv b/Y\uplus X \qquad
   \alpha \circ K_{\orivec{a}} \equiv K_{\alpha(\orivec{a})} \qquad
   (\text{$\alpha$ is a renaming})
\end{gather*}

To model properly the names (and hence the edges) of a link graphs we
have to introduce the fresh name quantifier $\reversen$, very similar
to the one defined in Nominal Logic.  The semantics of $\reversen$ is
defined as follows:
\begin{equation*}
  G \vDash \reversen x_{1} \ldots x_{n}.\; A \iff
  \exists  a_{1} \ldots a_{n} \notin fn(G) \cup fn(A). \;
  G \vDash A\{x_{1}/a_{1} \dots x_{n}/a_{n}\}
\end{equation*}

Using the quantification on fresh names is possible to define a
``separation'' that shares at most a specified set of names, e.g., as
the tensorial decomposition of two terms sharing the names
$\orivec{a}$:
\begin{multline*} 
  A \sepupto{\orivec{a}} B \defeq (\reversen \orivec{x}. \;
  ( a_1/\{x_1,a_1\} \otimes \dots \otimes a_n/\{x_n,a_n\}) \otimes id) \circ 
  \Big(
  \big((( x_1/a_1 \otimes \dots \otimes x_n/a_n ) \otimes id) \circ A\big)
  \otimes B
  \Big)
\end{multline*}
Notice that this predicate has been extensively used in the definition
of our sortings on the categories $\cat{Lgp}$ and $\cat{BioBg}$.

\smallskip
\noindent\textbf{Logic for Place Graphs} For instantiating the logic on
place graphs we consider the following set of constructor:
\begin{equation*}
  \Theta_{P} = \set{
    id_{n} \colon n \rightarrow n, \
    \mathbf{1} \colon 0 \rightarrow 1, \
    join\colon 2 \rightarrow 1, \
    \gamma_{1,1}\colon 2 \rightarrow 2}
\end{equation*}
The constructor $\mathbf{1}$ denotes the empty location.  $join$ maps
to location into one.  Finally, $\gamma_{1,1}$ swap the left location
with the right one.

As done in the case of link graphs, we extend the structural
congruence ($\equiv$) for the place graph case as follows:
\begin{gather*}
  join \circ (\mathbf{1} \otimes id_{1}) \equiv id_{1} \qquad \qquad
  join \circ \gamma_{1,1} \equiv join \\
  join \circ (join \otimes id_{1}) \equiv join \circ (id_{1} \otimes join)
\end{gather*}

\smallskip
\noindent\textbf{Logic for Bigraphs} It is defined by simply putting
together the constructors previously introduced, formally:
\begin{multline*}
  \Theta_{B} = \set{
    \mathbf{1}\colon \epsilon \rightarrow \langle 1,\emptyset \rangle,\,
    id_{\langle n,X\rangle} \colon
       \langle n,X\rangle \rightarrow \langle n,X\rangle,\,
    join \colon
       \langle 2,\emptyset \rangle \rightarrow \langle 1,\emptyset \rangle,
    \\
    \gamma_{1,1} \colon
       \langle 2,\emptyset \rangle \rightarrow \langle 2,\emptyset \rangle,\,
    /a{:\,}\mathsf{t} \colon \langle 0,\, \set{a}\rangle \rightarrow \epsilon,\,
    a/X \colon \langle 0,X\rangle \rightarrow \langle 0,\set{a}\rangle,\,
    K_{\orivec{a}} \colon \epsilon \rightarrow \langle 1,\orivec{a}\rangle
  }
\end{multline*}
	
The extension of the structural congruence ($\equiv$) is defined by
the ones introduce for place and link graphs, plus the following ones,
that recovers the properties of any symmetric category.
\begin{align*}
  \gamma_{I, \epsilon} \equiv id_{I} \qquad
  \gamma_{I,J} \circ \gamma_{J,I} \equiv id_{I \otimes J}  \qquad
  \gamma_{I', J'} \circ (G \otimes F) \equiv (F \otimes G) \circ \gamma_{I,J} 
\end{align*}

All the logic operators defined for the previous logics can be easily
lifted to the logic of bigraphs. (They must be ``lifted'' on the new
objects.)

\section{Sortings as BiLog formulae}\label{sec:formulae}

Sortings for bigraphs are studied in \cite{bdh:concur08}; in
particular, a general technique for defining a \emph{safe} sorting is
by means of a \emph{decomposable predicate} $P$, which specifies the
class of morphisms we want to restrict to \cite{bdh:concur08}.  A
predicate $P$ on morphisms is decomposable if $P(f \circ g)$ implies
$P(f)$ and $P(g)$.  Decomposable predicates can be conveniently
denoted by BiLog formulae of the form $(\neg \phi)^{\forall\circ}$
\cite{cms:icalp05}, where $\phi$ characterizes \emph{un}wanted
morphisms\footnote{The operator $(\cdot)^{\forall\circ}$ guarantees
  the predicate decomposition and so the resulting sorted category is
  well-defined.}.  Therefore, biological bigraphs can be characterized
by means of sortings specified by the corresponding BiLog formula.

In this section we will provide all the BiLog formulae for every
sorting used in our framework of Biological Bigraphs.

\smallskip\noindent
\ref{eq:protsol} states that a protein solution cannot have
ternary links and/or binary visible edges.
\begin{equation*}
  \ref{eq:protsol} \defeq
  \big(\neg 
  (K'_{\orivec{x}} \sepupto{a} K''_{\orivec{y}} \sepupto{a} K'''_{\orivec{z}}
  \otimes  \mathbf{T}) \big)^{\forall\circ}
  \wedge
  \big(\neg ((/a{:\,}\mathsf{v} \circ K'_{\orivec{x}} \sepupto{a} K''_{\orivec{y}}) \otimes
  \mathbf{T}) \big)^{\forall\circ}
  \label{eq:p2}
\end{equation*}

\smallskip\noindent
\ref{eq:imp} says that two nodes cannot be linked if there are more
than two occurrences of a membrane control along the path to the
common ancestor.
\begin{equation*}
 \ref{eq:imp} \defeq
 (\neg 
 (
 ( 
 \mext \circ (
 (id_{\set{a}} \lefttensoradj \mathbf{T}) 
 \mid (\mcis \circ \neg(a)) 
 ) 
 )
 \sepupto{a} \neg(a)))^{\forall\circ}.
\end{equation*}

\smallskip\noindent 
\ref{eq:polar}-\ref{eq:apolar} forbids the nesting of nodes into
locations of the same polarity:
\begin{align*}
  \ref{eq:polar} & \defeq
  \textstyle
  \Big( 
    \neg \big(
        (
	\neg\pcarry_{y} \vee \bigvee_{K \in \set{\mext}} K_{\orivec x}
        )
  	\circ \bigvee_{K \in \mathcal{P}\uplus\set{\mext}} K_{\orivec z}
    \big)
  \Big)^{\forall\circ}
  \\
  \ref{eq:apolar} & \defeq
  \textstyle
  \Big( 
    \neg \big(
        (
	\neg\pmemb_{y} \vee \bigvee_{K \in \set{\mcis}} K_{\orivec x}
        )
  	\circ \bigvee_{K \in \mathcal{A}\uplus\set{\mcis}} K_{\orivec z}
     \big)
  \Big)^{\forall\circ}
\end{align*}

\smallskip\noindent
\ref{eq:bi-layer} enforces the well-nesting of cells.
\begin{equation*}
  \ref{eq:bi-layer} \defeq
  \Big(
  \neg \big(
  \mext \circ (
  ( \mcis \mid \mcis \mid \mathbf{id}_1 )
  \vee
  \neg( \mcis \mid \mathbf{id}_1)
  )
  \big)
  \Big)^{\forall\circ}
\end{equation*}

\smallskip\noindent
\ref{eq:mob} forbids that a mobility node can be peered with a protein
node, forcing it to be linked with another mobility node of the same
type ($P$ or $F$).
\begin{align*}
  \ref{eq:mob} \defeq {} & 
  \Big( \textstyle
  \neg(
  ( \bigvee_{
  		\substack{K \in \mathcal{P}\cup\mathcal{A} \\ 
  		K' \notin\mathcal{P}\cup\mathcal{A}}
		} 
	K_{\orivec{x}} \sepupto{a} K'_{\orivec{y}} 
  )
  )
  \wedge
  \neg(
  ( \bigvee_{\substack{M \in \mathcal{M} \\ t \in \set{\mathsf{h}, \mathsf{v}} }} /z{:\,}t \circ M_{\orivec{z}} 
  )
  ) \wedge
  \neg(
  (\bigvee_{\substack{P \in \mathcal{M}_{P} \\ 
      F \in \mathcal{M}_{F} }} 
  P_{\orivec{w}} \sepupto{b} F_{\orivec{k}} 
  )
  \Big)^{\forall\circ}
  \wedge {} \\
  \textstyle
  &\Big(
  	\neg \big(
		\bigvee_{M \in \mathcal{M}} M_{\orivec{x}} \sepupto{a} M_{\orivec{x}}
	\big)
	\wedge
	\neg \big(
		\bigvee_{\mathsf{t} \in \set{\mathsf{p}, \mathsf{f}}}
			(\mathsf{t}^{\mathsf{c}}_{a} \sepupto{a} \mathsf{t}^{\mathsf{d}}_{a}) 
			\vee
			(\mathsf{t}^{\mathsf{m}}_{xa} \sepupto{a} \mathsf{t}^{\mathsf{c}}_{a})
			\vee
			(\mathsf{t}^{\mathsf{m}}_{ay} \sepupto{a} \mathsf{t}^{\mathsf{d}}_{a})
	\big)
  \Big)^{\forall\circ}
\end{align*}

\smallskip\noindent
\ref{eq:dir} disallows mobility controls $\pcarry$, $\pdir$,
$\fcell$, $\fdir$ to reside in the same location.
\begin{equation*}
  \ref{eq:dir} \defeq {}
  \big(\neg(
  join \circ (\pcarry_{x}\otimes \pdir_{y}) \parallel \pmemb_{xy}
  )
  \big)^{\forall\circ}
  \wedge
  \big(\neg(
  join \circ (\fcell_{x}\otimes \fdir_{y}) \parallel \fmemb_{xy}
  )
  \big)^{\forall\circ}
\end{equation*}

\smallskip\noindent
\ref{eq:nonest} ensures the mobility actions are not nested.
\begin{equation*}
  \ref{eq:nonest} \defeq 
  \big( \textstyle
  (\bigvee_{m\in \set{\pcarry,\pmemb,\fcell}} m_{\vec{x}}) \circ \mathbf{T} \circ 
  (\bigvee_{m\in \mathcal{M}} m_{\vec{y}})
  \big)^{\forall\circ}
\end{equation*}

\smallskip\noindent
\ref{eq:fuse} ensures well-defined fusions.
\begin{equation*}
  \ref{eq:fuse} \defeq
  \big ( \textstyle 
  \neg ( \fcell \circ (\mext \otimes \mext \otimes \mathbf{T}) \vee
         \fcell \circ (\bigvee_{K\in \mathcal{P} \cup \mathcal{A}} K_{\vec{x}}
                       \otimes \mathbf{T}))
  \big )^{\forall\circ}
\end{equation*}

\end{extended}

\end{document}